%% file: ijcai21-multiauthor.tex
\pdfoutput=1

\documentclass{article}
\usepackage{ijcai21}

\usepackage{xr-hyper}

\makeatletter
\newcommand*{\addFileDependency}[1]{
  \typeout{(#1)}
  \@addtofilelist{#1}
  \IfFileExists{#1}{}{\typeout{No file #1.}}
}
\makeatother

\newcommand{\dep}{\textsc{dep}}

\usepackage{times}

\usepackage{soul}
\usepackage{url}
\usepackage[hidelinks]{hyperref}
\usepackage[utf8]{inputenc}
\usepackage[small]{caption}
\usepackage{graphicx}
\usepackage{amsmath}
\usepackage{booktabs}
\usepackage{xcolor}

\usepackage{wrapfig}

\usepackage{latexsym}

\usepackage{algorithm}
\usepackage[noend]{algpseudocode}

\usepackage{amsfonts}
\usepackage{amssymb}
\usepackage{amsthm}
\usepackage{multirow}
\usepackage{subcaption}

\newtheorem{theorem}{Theorem}[section]
\newtheorem{definition}{Definition}[section]
\newtheorem{lemma}[theorem]{Lemma}
\newtheorem{corollary}[theorem]{Corollary}

\newcommand{\yw}[1]{\textcolor{orange}{\bf\small [yi: #1]}}

\DeclareMathOperator{\eu}{V}

\title{Temporal Induced Self-Play for Stochastic Bayesian Games}

\author{
Weizhe Chen$^1$\footnote{Equal Contribution}\and
Zihan Zhou$^2$\footnotemark[\value{footnote}]\and
Yi Wu$^{2, 3}$\And
Fei Fang$^4$\\
\affiliations
$^1$Shanghai Jiao Tong University\\
$^2$Shanghai Qi Zhi Institute\\
$^3$Tsinghua University\\
$^4$Carnegie Mellon University\\
\emails
chenweizhe@sjtu.edu.cn,
\{footoredo, jxwuyi\}@gmail.com,
feif@cs.cmu.edu
}

\begin{document}

\maketitle

\begin{abstract}
One practical requirement in solving dynamic games is to ensure that the players play well from any decision point onward. To satisfy this requirement, existing efforts focus on equilibrium refinement, but the scalability and applicability of existing techniques are limited.
In this paper, we propose Temporal-Induced Self-Play (TISP), a novel reinforcement learning-based framework to find strategies with decent performances from any decision point onward. TISP uses belief-space representation, backward induction, policy learning, and non-parametric approximation.
Building upon TISP, we design a policy-gradient-based algorithm TISP-PG. 
We prove that TISP-based algorithms can find approximate Perfect Bayesian Equilibrium in zero-sum one-sided stochastic Bayesian games with finite horizon. We test TISP-based algorithms in various games, including finitely repeated security games and a grid-world game. The results show that TISP-PG is more scalable than existing mathematical programming-based methods and significantly outperforms other learning-based methods.
\end{abstract}

\input{01Intro}

\input{02RelatedWork}

\input{03Preliminaries}

\input{04TINS}

\input{06Experiment}

\input{07Discussion}

\section*{Acknowledgements}
We would like to thank Ryan Shi for some help in writing the early workshop version of this paper.
Co-author Fang is supported in part by NSF grant IIS- 1850477, a research grant from Lockheed Martin, and by the U.S. Army Combat Capabilities Development Command Army Research Laboratory Cooperative Agreement Number W911NF-13-2-0045 (ARL Cyber Security CRA). The views and conclusions contained in this document are those of the authors and should not be interpreted as representing the official policies, either expressed or implied, of the funding agencies.


\bibliographystyle{named}
\bibliography{ijcai21-multiauthor}

\newpage
\input{appendix}

\end{document}

%% file: 01Intro.tex
\section{Introduction}
Many real-world problems involve multiple decision-makers interacting strategically. Over the years, a significant amount of work has focused on building game models for these problems and designing computationally efficient algorithms to solve the games~\cite{serrino2019finding,nguyen2019deception}. 
While Nash equilibrium (NE) is a well-accepted solution concept, the players' behavior prescribed by an NE can be irrational off the equilibrium path: one player can threaten to play a suboptimal action in a future decision point to convince the other players that they would not gain from unilateral deviation. Such ``non-credible threats'' restrict the practical applicability of these strategies as in the real world, one may make mistakes unexpectedly, and it is hard to enforce such threats.
Thus it is important to find strategy profiles such that each player's strategy is close to optimal (in expectation) from any point onward given the other players' strategies.

To find such strategy profiles, researchers have proposed equilibrium refinement concepts such as subgame perfect equilibrium and perfect Bayesian equilibrium (PBE)~\cite{cho1987signaling}
and studied the computational complexity~\cite{an2011refinement,etessami2014complexity,hansen2018computational}.
However, existing methods for computing refined equilibria have limited scalability and often require full access to the game environment, thus can hardly apply to complex games and real-world problems (as detailed in Section~\ref{sec:relatedwork}).
On the other hand, deep reinforcement learning (RL) has shown great promise in complex sequential decision-making problems for single-agent and multi-agent settings~\cite{mnih2015human,silver2018general}.
Deep RL leverages a compact representation of the game's state and the players' action space, making it possible to handle large games that are intractable for non-learning-based methods.
Despite the promise, to our knowledge, no prior work has applied deep RL to \textit{equilibrium refinements}.
In this paper, we focus on two-player stochastic Bayesian games with finite horizon as they can be used to model various long-term strategic interactions with private information~\cite{albrecht2013game}. We propose Temporal-Induced Self-Play (TISP), the first RL-based framework to find strategy profiles with decent performances from any decision point onward. 
There are several crucial challenges in using RL for this task. 
First, in these games, a player's action at a decision point should be dependent on the entire history of states and joint actions. As the number of histories grows exponentially, a tabular approach that enumerates all the histories is intractable. Although recurrent neural networks (RNNs) can be used to encode the history, RNNs are typically brittle in training and often fail to capture long-term dependency in complex games. Second, using standard RL algorithms with self-play suffers from limited exploration. Hence, it is extremely hard to improve the performance on rarely visited decision points. Our framework TISP tackles these two challenges jointly. We use a belief-based representation to address the first challenge, so that the policy representation remains constant in size regardless of the number of rounds. Besides, we use backward induction to ensure exploration in training. TISP also uses non-parametric approximation in the belief space. 
Building upon TISP, we design TISP-PG approach that uses policy gradient (PG) for policy learning. TISP can also be combined with other game-solving techniques such as counterfactual regret minimization (CFR)~\cite{zinkevich2007regret}. Further, we prove that TISP-based algorithms can find approximate PBE in zero-sum stochastic Bayesian games with one-sided incomplete information and finite horizon.  We evaluate TISP-based algorithms in different games. We first test them in finitely repeated security games with unknown attacker types whose PBE can be approximated through mathematical programming (MP) under certain conditions~\cite{nguyen2019deception}. Results show that our algorithms can scale to larger games and apply to more general game settings, and the solution quality is much better than other learning-based approaches. We also test the algorithms in a two-step matrix game with a closed-form PBE. Our algorithms can find close-to-equilibrium strategies. Lastly, we test the algorithms in a grid-world game, and the experimental results show that TISP-PG performs significantly better than other methods.

%% file: 02RelatedWork.tex
\section{Related Work}

\label{sec:relatedwork}
The study of equilibrium refinements is not new in economics~\cite{kreps1982sequential}.
In addition to the backward induction method for perfect information games, mathematical programming (MP)-based methods~\cite{nguyen2019deception,farina2017extensive,miltersen2010computing} have been proposed to compute refined equilibria.
However, the MPs used are non-linear and often have an exponential number of variables or constraints, resulting in limited scalability. A few works use iterative methods~\cite{farina2017regret,kroer2017smoothing} but they require exponentiation in game tree traversal and full access to the game structure, which limits their applicability to large complex games.

Stochastic Bayesian games have been extensively studied in mathematics and economics~\cite{forges1992repeated,sorin2003stochastic,chandrasekaran2017markov,albrecht2013game}. 
\cite{albrecht2016belief} discussed the advantage of using type approximation to approximate the behaviors of agents to what have already been trained, to reduce the complexity in artificial intelligence (AI) researches.
We focus on equilibrium refinement in these games and provide an RL-based framework.

Various classical multi-agent RL algorithms~\cite{littman1994markov,hu1998multiagent} are guaranteed to converge to an NE. 
Recent variants~\cite{heinrich2015fictitious,lowe2017multi,iqbal2019actor} leverage the advances in deep learning~\cite{mnih2015human} and have been empirically shown to find well-performing strategies in  large-scale games, such as Go~\cite{silver2018general} and StarCraft~\cite{vinyals2019grandmaster}.
We present an RL-based approach for equilibrium refinement. 

Algorithms for solving large zero-sum imperfect information games like Poker~\cite{moravvcik2017deepstack,brown2018superhuman} need to explicitly reason about beliefs. 
Many recent algorithms in multi-agent RL use belief space policy or reason about joint beliefs. These works assume a fixed set of opponent policies that are unchanged~\cite{shen2019robust}, or consider specific problem domains~\cite{serrino2019finding,Woodward_Finn_Hausman_2020}. Foerster \emph{et al.}\shortcite{foerster2019bayesian} uses public belief state to find strategies in a fully cooperative partial information game. We consider stochastic Bayesian games and use belief over opponents' types.


%% file: 03Preliminaries.tex
\section{Preliminaries}

\subsection{One-Sided Stochastic Bayesian Game}
For expository purposes, we will mainly focus on what we call one-sided stochastic Bayesian games (OSSBG), which extends finite-horizon two-player stochastic games with type information.
In particular, player 1 has a private \textit{type} that affects the payoff function. Hence, in a competitive setting, player 1 needs to hide this information, while player 2 needs to infer the type from player 1's actions. Our algorithms can be extended to handle games where both players have types, as we will discuss in Section \ref{sec:conclusion}. 
Formally, an OSSBG is defined by a 8-tuple $\Gamma=\langle\Omega, \mu^0, \Lambda, p^0, \mathcal{A}, P, \{u_i^\lambda\},L, \gamma\rangle$. 
$\Omega$ is the state space. 
$\mu^0$ is the initial state distribution.
$\Lambda=\{1,\dots, |\Lambda|\}$ is the set of types for player 1.
$p^0$ is the prior  over player 1's type.
$\mathcal{A}=\prod_i\mathcal{A}_i$ is the joint action space with $\mathcal{A}_i$ the action space for player $i$. 
$P: \Omega \times \mathcal{A} \rightarrow \Delta_{|\Omega|}$ is the transition function where $\Delta_{|\Omega|}$ represents the $|\Omega|$-dimensional probability simplex. 
$u_i^{\lambda}: \Omega \times \mathcal{A} \rightarrow \mathbb{R}$ is the payoff function for player $i$ given player 1's type $\lambda$.
$L$ denotes the length of the horizon or number of rounds.
$\gamma$ is the discount factor.

One play of an OSSBG starts with a type $\lambda$ sampled from $p^0$ and an initial state  $s^0$ sampled from $\mu^0$. Then, $L$ \textit{rounds} of the game will rollout. In round $l$, players take actions $a^l_1\in\mathcal{A}_1$ and $a^l_2\in\mathcal{A}_2$ simultaneously and independently, based on the history $h^l:=\{s^0,(a_1^0,a_2^0),\dots,(a_1^{l-1},a_2^{l-1}),s^l\}$. The players will then get payoff $u^\lambda_i(s^t,a^l_1,a^l_2)$. Note that the payoff $u_i^{\lambda}(s^l, a_1, a_2)$ at every round $l$ will not be revealed until the end of every play on the game to prevent type information leakage. The states transit w.r.t. $P(s^{l+1}|s^l,a^l_1,a^l_2)$ across rounds.

Let $\mathcal{H}^l$ denote the set of all possible histories in round $l$ and $\mathcal{H}=\bigcup_l\mathcal{H}^l$. Let $\pi_1:\Lambda\times\mathcal{H}\rightarrow \Delta_{\mathcal{A}_1}, \pi_2:\mathcal{H}\rightarrow \Delta_{\mathcal{A}_2}$ be the players'  behavioral strategies. Given the type $\lambda$, the history $h^l$ and the strategy profile $\pi=(\pi_1,\pi_2)$, player $i$'s discounted accumulative expected utility from round $l$ onward is 
\begin{align}
    &\eu_i^{\lambda}(\pi, h^l) =  \sum_{a_1, a_2} 
    \bigg( \pi_1(a_1|\lambda, h^l) \pi_2(a_2|h^l) \Big(u_i^{\lambda}(s^l, a_1, a_2) \nonumber \\
    &\, + \gamma \sum_{s'}P(s'|s^l, a_1, a_2) \eu_i^{\lambda}(\pi,h^l \cup \{(a_1, a_2), s'\})\Big) \bigg). \label{eq:eu}
\end{align}
Similarly, we can define the Q function by 
$Q_i^\lambda(\pi,h^l,a)=u_i^\lambda(s^l,a_1,a_2)+\gamma\mathbb{E}_{s'}\left[V_i^\lambda(\pi,h^l\cup\{(a_1,a_2),s'\})\right]$.

An OSSBG can be converted into an equivalent extensive-form game (EFG) with imperfect information where each node in the game tree corresponds to a (type, history) pair (see Appendix~\ref{app-sec:main-proof}). 
However, this EFG is exponentially large, and existing methods for equilibrium refinement in EFGs~\cite{kroer2017smoothing} are not suitable due to their limited scalability.

\subsection{Equilibrium Concepts in OSSBG} \label{sec:equilibrium_concepts}
Let $\Pi_i$ denote the space of all valid strategies for Player $i$.

\begin{definition}[$\epsilon$-NE]
	 A strategy profile $\pi=(\pi_1, \pi_2)$ is an $\epsilon$-NE if for $i=1,2$, and all visitable history $h^l$ corresponding to the final policy,
	 \begin{equation}
	 \max_{\pi'_i\in\Pi_i}\eu_i(\pi_i|_{h^l\rightarrow \pi'_i},\pi_{-i}, h^l)-\eu_i(\pi, h^l)\leq \epsilon
	 \end{equation}
	 where $\pi_i|_{h^l\rightarrow \pi'_i}$ means playing $\pi_i$ until $h^l$ is reached, then playing $\pi'_i$ onwards.
\end{definition}

\begin{definition}[$\epsilon$-PBE]
	 A strategy profile $\pi=(\pi_1, \pi_2)$ is an $\epsilon$-PBE if for $i=1,2$ and all histories $h^l$,
	\begin{equation}
		\max_{\pi'_i\in\Pi_i}\eu_i(\pi_i|_{h^l\rightarrow \pi'_i},\pi_{-i},h^l)-\eu_i(\pi,h^l)\leq \epsilon
	\end{equation}
	where $\pi_i|_{h^l\rightarrow \pi'_i}$ means playing $\pi_i$ until $h^l$ is reached, then playing $\pi'_i$ onwards.
\end{definition}
It is straightforward that an $\epsilon$-PBE is also an $\epsilon$-NE. 

%% file: 04TINS.tex
\section{Temporal-Induced Self-Play}



Our TISP framework (Alg.~\ref{alg:ti_framework}) considers each player as an RL agent and trains them with self-play. Each agent maintains a policy and an estimated value function, which will be updated during training.
TISP has four ingredients. It uses belief-based representation (Sec.~\ref{sec:belief_algo}), backward induction (Sec.~\ref{sec:bi}), policy learning (Sec.~\ref{sec:policy}) and belief-space approximation (Sec.~\ref{sec:approx}).
We discuss test-time strategy and show that TISP converges to $\epsilon$-PBEs under certain conditions in Sec. ~\ref{sec:test-time}. 

\subsection{Belief-Based Representation}\label{sec:belief_algo}

Instead of representing a policy as a function of the history, we consider player 2's belief $b\in \Delta_{|\Lambda|}$ of player 1's type and represent $\pi_i$ as a function of the belief $b$ and the current state $s^l$ in round $l$, i.e., $\pi_{1,l}(\cdot|b,s^l,\lambda)$ and $\pi_{2,l}(\cdot|b,s^l)$.
The belief $b$ represents the posterior probability distribution of $\lambda$ and can be obtained using Bayes rule given player 1's strategy:
\begin{align}
    b_{\lambda}^{l+1}=\frac{\pi_{1,l}\left(a_1^l|s^l,b_\lambda^l,\lambda\right)b^l_{\lambda}}{\sum_{\lambda' \in \Lambda}\pi_{1,l}\left( a_1^l|s^l,b_{\lambda'}^l,\lambda' \right) b^l_{\lambda'}}\label{eq:next_belief}
\end{align}
where $b^l_\lambda$ is the probability of player 1 being type $\lambda$ given all its actions up to round $l-1$.
This belief-based representation avoids the enumeration of the exponentially many histories.
Although it requires training a policy that outputs an action for any input belief in the continuous space, it is possible to use approximation as we show in Sec.~\ref{sec:approx}.
We can also define the belief-based value function for agent $i$ in round $l$ by 
{
\begin{align}
V_{i,l}^\lambda(\pi,b^l,s^l)=\mathbb{E}_{a_1,a_2,s^{l+1}}& \big[u_i^\lambda(s^l,a_1,a_2) \nonumber \\
& +\gamma V_{i,l+1}^\lambda(\pi,b^{l+1},s^{l+1})\big].\label{eq:belief_V}
\end{align}
}
The Q-function $Q_{i,l}^\lambda(\pi,b^l,s^l,a^l)$ can be defined similarly. We assume the policies and value functions are parameterized by $\theta$ and $\phi$ respectively with neural networks.

\subsection{Backward Induction}\label{sec:bi}

Standard RL approaches with self-play train policies in a top-down manner: it executes the learning policies from round $0$ to $L-1$ and only learns from the experiences at visited decision points. To find strategy profiles with decent performances from any decision point onward, we use backward induction and train the policies and calculate the value functions in the reverse order of rounds: we start by training $\pi_{i,L-1}$ for all agents and then calculate the corresponding value functions $V_{i,L-1}^\lambda$, and then train $\pi_{i,L-2}$ and so on. 

The benefit of using backward induction is two-fold. In the standard forward-fashioned approach, one needs to roll out the entire trajectory to estimate the accumulative reward for policy and value learning. In contrast, with backward induction, when training $\pi_{i,l}$, we have already obtained $V_{i,l+1}^\lambda$. Thus, 
we just need to roll out the policy for 1 round and directly estimate the expected accumulated value using $V_{i,l+1}^\lambda$ and Eq. (\ref{eq:belief_V}). Hence, we effectively reduce the original $L$-round game into $L$  1-round games, which makes the learning much easier. 
Another important benefit is that we can uniformly sample all possible combinations of state, belief and type at each round to ensure effective exploration. 
More specifically, in round $l$, we can sample a belief $b$ and then construct a new game by resetting the environment with a uniformly randomly sampled state and a type sampled from $b$. 
Implementation-wise, we assume access to an auxiliary function from the environment, called $\textit{sub\_reset}(l,b)$ that produces a new game as described. 
This function takes two parameters, a round $l$ and a belief $b$, as input and produces a new game by drawing a random state from the entire state space $\Omega$ with equal probability and a random type according to the belief distribution $b$.
This function is an algorithmic requirement for the environment, which is typically feasible in practice. For example, most RL environments provides a \textit{reset} function that generates a random starting state, so a simple code-level enhancement on this \textit{reset} function can make existing testbeds compatible with our algorithm.
We remark that even with such a minimal environment enhancement requirement, our framework does \emph{NOT} utilize the transition information. Hence, our method remains nearly model-free comparing to other methods that assume full access to the underlying environment transitions --- this is the assumption of most CFR-based algorithms. Furthermore, using customized reset function is not rare in the RL literature. For example, most automatic curriculum learning algorithms assumes a flexible reset function that can reset the state of an environment to a desired configuration.

\subsection{Policy Learning}\label{sec:policy}
Each time a game is produced by $\textit{sub\_reset}(l,b)$, we perform a 1-round learning with self-play to find the policies $\pi_{i,l}$. TISP allows different policy learning methods. Here we consider two popular choices, policy gradient and regret matching.


\subsubsection{Policy Gradient}
PG method directly takes a gradient step over the expected utility. For notational conciseness, we omit the super/subscripts of $i$, $l$, $\lambda$ in the following equations and use $s$, $b$ and $s'$, $b'$ to denote the state and belief at current round and next round. 

\begin{theorem}
In the belief-based representation, the policy gradient derives as follows:

{
\begin{align}
    \nabla_{\theta} V^{\lambda}(\pi, b, s)  =  \sum_{a \in \mathcal{A}}  \nabla_{\theta} & \left(\pi_\theta (a \vert b, s) Q^\lambda(\pi, b, s, a) \right) \label{eq:pg_update} \\
    = \mathbb{E}_{a, s'} \Big[ Q^\lambda(\pi, b, s, a) & \nabla_\theta \ln \pi_\theta(a \vert b, s) \nonumber\\
    & + \gamma \nabla_\theta b'\nabla_{b'}V^\lambda(\pi,b', s') \Big].\nonumber 
\end{align}
}
\end{theorem}

Comparing to the standard policy gradient theorem, we have an additional term in Eq.(\ref{eq:pg_update}) (the second term). Intuitively, when the belief space is introduced, the next belief $b'$ is a function of the current belief $b$ and the policy $\pi_{\theta}$ in the current round (Eq.(\ref{eq:next_belief})). Thus, the change in the current policy may influence future beliefs, resulting in the second term.
The full derivation can be found in Appendix~\ref{app-sec:pg}. We also show in the experiment section that when the second term is ignored, the learned policies can be substantially worse. We refer to this PG
variant of TISP framework as TISP-PG.

\subsubsection{Regret Matching}


Regret matching is another popular choice for imperfect information games. We take inspirations from Deep CFR~\cite{Brown2019DeepCR} and propose another variant of TISP, referred to as TISP-CFR. Specifically, for each training iteration $t$, let $R^t(s, a)$ denote the regret of action $a$ at state $s$, $\pi^t(a|s,b)$ denote the current policy, $Q^t(\pi^t,s,b,a)$ denote the Q-function and $V^t_\phi(\pi^t,s,b)$ denote the value function corresponding to $\pi^t$. Then we have  
\begin{eqnarray*}
	\pi^{t+1}(a|s,b)&=&\frac{(R^{t+1}(s, b, a))^+}{\sum_{a'}(R^{t+1}(s, b, a'))^+},
\end{eqnarray*}
where $R^{t+1}(s, b, a)=\sum_{\tau=1}^t Q^\tau(\pi^\tau, s, b, a)-V^\tau_\phi(\pi^\tau, s, b)$ and $(\cdot)^+:=\max(\cdot, 0)$. Since the policy can be directly computed from the value function, we only learn $V^t_\phi$ here.
Besides, most regret-based algorithms require known transition functions, which enables them to reset to any infomation-set node in the training. We use the outcome sampling method~\cite{lanctot2009monte}, which samples a batch of transitions to update the regret. This is similar to the value learning procedure in standard RL algorithms. This ensures our algorithm to be model-free. Although TISP-CFR does not require any gradient computation, it is in practice much slower than TISP-PG since it has to learn an entire value network in every single iteration. 

\begin{algorithm}[t]
\caption{Temporal-Induced Self-Play}
\label{alg:ti_framework}

\begin{algorithmic}[1]
\For {$l=L - 1, \dots, 0$}
    \For {$k=1, 2, \dots K$} \Comment{ run in parallel}
        \For {$t=1, \dots, T$}
            \State Initialize replay buffer $D=\{\}$ and $\pi^0$
            \For {$j=1, \dots, \textrm{batch size}$} \Comment{parallel}
            \State $s \gets sub\_reset(l, b_k)$;
            \State $a \gets  \pi^{t-1}_{\theta_{l,k}}(s;b_k)$;
            
            \State get next state $s'$ and utility $u$ from env;
            \State $D \gets D + (s,a,s',u)$;
            \EndFor
            \State Update ${V^t_{\phi_{l,k}}}$ and $\pi^t_{\theta_{l,k}}$  using $D$;
        \EndFor
        \State $V_{\phi_{l,k}}\gets V^{n}_{\phi_{l,k}}$, $\pi_{\theta_{l,k}}\gets \pi^n_{\theta_{l,k}}$
    \EndFor
\EndFor
\State \Return $\{\pi_{\theta_{l,k}}, V_{\phi_{l,k}}\}_{0\le l <L, 1\le k \le K}$;
\end{algorithmic}
\end{algorithm}

\subsection{Belief-Space Policy Approximation} \label{sec:approx}

A small change to the belief can drastically change the policy. When using a function approximator for the policy, this requires the network to be sensitive to the belief input. 
We empirically observe that a single neural network often fails to capture the desired belief-conditioned policy. Therefore, we use an ensemble-based approach to tackle this issue.


At each round, we sample $K$ belief points $\{ b_1, b_2, ..., b_K \}$ from the belief space $\Delta^{|\Lambda|}$. For each belief $b_k$, we use self-play to learn an accurate independent strategy $\pi_{\theta_{k}}(a|s;b_k)$ and value $V_{\phi_{k}}(\pi,s;b_k)$ over the state space but specifically conditioning on this particular belief input $b_k$.
When querying the policy and value for an arbitrary belief $b$ different from the sampled ones, we use a distance-based non-parametric method to approximate the target policy and value. Specifically, for any two belief $b_1$ and $b_2$, we define a distance metric $w(b_1,b_2)= \frac{1}{\max \left( \epsilon, \|b-b'\|^2 \right)} $ and then for the query belief $b$, we calculate its policy $\pi(a | b, s)$ and value $V(\pi, b, s)$ by
\begin{eqnarray}
    \pi(a | b, s) &=& \frac{\sum_{k=1}^K \pi_{\theta_k}(a |s; b_k)w(b, b_k)}{\sum_{k=1}^K w(b, b_k)}\label{eq:approx-pi} \\
    V(\pi, b, s)&=&\frac{\sum_{k=1}^K V_{\phi_k}(\pi, s; b_k) w(b, b_k)}{\sum_{k=1}^K w(b, b_k)}
\end{eqnarray}





We introduce a density parameter $d$ to ensure the sampled points are dense and has good coverage over the belief space. $d$ is defined as the farthest distance between any point in $\Delta^{|\Lambda|}$ and its nearest sampled points, or formally $d=\sup\{\min\{\|b, b_i\|\mid i=1,\dots,K\}\mid b\in \Delta^{|\Lambda|}\}$. A smaller $d$ indicates a denser sampled points set.

Note that the policy and value networks at different belief points at each round are completely independent and can be therefore trained in perfect parallel. So the overall training wall-clock time remains unchanged with policy ensembles. Additional discussions on belief sampling are in  Appx.~\ref{app-sec:belief-sampling}.

\subsection{Test-time Strategy}\label{sec:test-time}

\begin{algorithm}[tb]
\caption{Compute Test-Time Strategy}\label{alg:convert}

\begin{algorithmic}[1]
	\Function{GetStrategy}{$h^l, \pi_{\theta_1},\ldots,\pi_{\theta_L}$}
		\State $b^0\gets$ $p^0$
		\For{$j\gets 0,\dots, l-1$}
			\State update $b^{j+1}$ using $b^j$, $s^j$, $a^j$ and $\pi_{\theta_j}$   with Eq.(\ref{eq:next_belief})
		\EndFor
		\State \Return $\pi(a|b^{l},s^l)$ with Eq.(\ref{eq:approx-pi})
	\EndFunction
\end{algorithmic}
	
\end{algorithm}


Note that our algorithm requires computing the precise belief using Eq.~\ref{eq:next_belief}, which requires the known policy for player 1, which might not be feasible at test time when competing against an unknown opponent. 
Therefore, at test time, we update the belief according to the \emph{training policy} of player 1, regardless of its actual opponent policy. That is, even though the actions produced by the actual opponent can be completely different from the oracle strategies we learned from training, we still use the oracle policy from training to compute the belief. Note that it is possible that we obtain an infeasible belief, i.e., the opponent chooses an action with zero probability in the oracle strategy. In this case, we simply use a uniform belief instead. The procedure is summarized in Alg.~\ref{alg:convert}. We also theoretically prove in Thm.~\ref{thm:main} that in the zero-sum case, the strategy provided in Alg.~\ref{alg:convert} provides a bounded approximate PBE and further converges to a PBE with infinite policy learning iterations and sampled beliefs. The poof can be found in Appendix~\ref{app-sec:main-proof}.



\begin{theorem}\label{thm:main}
	  When the game is zero-sum and $\Omega$ is finite, the strategies produced by Alg.~\ref{alg:convert} is $\epsilon$-PBE, where $\epsilon=L(dU+c\cdot T^{-1/2})$, $d$ is the distance parameter in belief sampling, $U=L\cdot\big(\max_{i,s,a_1,a_2} u_i(s,a_1,a_2)-\min_{i,s,a_1,a_2} u_i(s,a_1,a_2)\big)$, $n$ is the number of iterations in policy learning and $c$ is a positive constant associated with the particular algorithm (TISP-PG or TISP-CFR, details in Appendix~\ref{app-sec:main-proof}). When $T\rightarrow \infty$ and $d\rightarrow 0$, the strategy becomes a PBE.
\end{theorem}

We remark that TISP is nearly-model-free and does not utilize the transition probabilities, which ensures its adaptability.



%% file: 06Experiment.tex
\section{Experiment}
We test our algorithms in three sets of games. While very few previous works in RL focus on equilibrium refinement, we compare our algorithm with self-play PG with RNN-based policy (referred to as RNN) and provide ablation studies for the TISP-PG algorithm: \emph{BPG} uses only belief-based policy without backward induction or belief-space approximation; \emph{BI} adopt backward induction and belief-based policy but does not use belief-space approximation. 
Full experiment details can be found in Appx.~\ref{app-sec:expr}. 


\subsection{Finitely Repeated Security Game}

\subsubsection{Game Setting}
We consider a finitely repeated \textit{simultaneous-move} security game, as discussed in \cite{nguyen2019deception}. Specifically, this is an extension of a one-round security game by repeating it for $L$ rounds. Each round's utility function can be seen as a special form of matrix game and remains the same across rounds. In each round, the attacker can choose to attack one position from all $A$ positions, and the defender can choose to defend one. The attacker succeeds if the target is not defended. The attacker will get a reward if it successfully attacks the target and a penalty if it fails. Correspondingly, the defender gets a penalty if it fails to defend a place and a reward otherwise. In the zero-sum setting, the payoff of the defender is the negative of the attacker's. We also adopt a general-sum setting described in \cite{nguyen2019deception} where the defender's payoff is only related to the action it chooses, regardless of the attacker's type. 



\subsubsection{Evaluation}
We evaluate our solution by calculating the minimum $\epsilon$ so that our solution is an $\epsilon$-PBE. We show the average result of 5 different game instances. 

\begin{table}[t]
\begin{subtable}[h]{0.48\textwidth}
\small
\centering
\begin{tabular}{l|llll}
{}    & TISP-PG & RNN & BPG & BI  \\ \hline
Mean $\epsilon$&   \textbf{0.881}  &  15.18     &  101.2  &  27.51  \\ \hline
Worst $\epsilon$  &   \textbf{1.220} &   31.81    &  111.8  & 42.54
\end{tabular}
\vspace{-5pt}
\caption{Zero-sum result for model-free methods, with $|\Lambda|=2$.}
\label{table:zero_sum_sec_game}
\end{subtable}
\begin{subtable}[h]{0.48\textwidth}
\small
\centering
\begin{tabular}{l|llll}
{}    & TISP-PG & RNN & BPG & BI  \\ \hline
Mean $\epsilon$&   \textbf{0.892}  &  34.62     &  89.21  &  83.00  \\ \hline
Worst $\epsilon$  &   \textbf{1.120} &   57.14    &  182.1  & 111.9  
\end{tabular}
\vspace{-5pt}
\caption{General-sum result for model-free methods, with $|\Lambda|=2$.}
\label{table:general_sum_sec_game}
\end{subtable}
\begin{subtable}[h]{0.48\textwidth}
\small
\centering
\begin{tabular}{l|ll|ll}
{} & \multicolumn{2}{c|}{Zero-sum} & \multicolumn{2}{c}{General-sum} \\ \cline{2-5}
{}    &  TISP-PG & TISP-CFR  &  TISP-PG & TISP-CFR \\ \hline
Mean $\epsilon$ &   \textbf{0.446}  &   0.474 &  \textbf{0.608}  &   0.625 \\ \hline
Worst $\epsilon$    &   \textbf{1.041}  & 1.186 & \textbf{1.855}  & 1.985
\end{tabular}
\vspace{-5pt}
\caption{Result for known model variants, with $|\Lambda|=2$.}
\label{table:known_model_sec_game}
\end{subtable}
\begin{subtable}[h]{0.45\textwidth}
\small
\centering
\begin{tabular}{l|llll}
{}    & TISP-PG & RNN & BPG & BI \\ \hline
Mean $\epsilon$&   \textbf{1.888}  &   18.20    &  79.74  & 40.75   \\ \hline
Worst $\epsilon$  &   \textbf{3.008} &    28.15   &  97.67  & 49.74
\end{tabular}
\caption{Zero-sum result, with $|\Lambda|=3$.}
\label{table:3types_sec_game}
\end{subtable}
\vspace{-7pt}
\caption{The result for finitely repeated security game. 
The less the number, the better the solution is. These results are evaluated with $L=10$, $|A|=2$, and uniform prior distribution.}
\label{table:sec_game}
\end{table}


\begin{table}[t]
\begin{subtable}[h]{0.45\textwidth}
\small
\centering
\begin{tabular}{l|lllll}
L        & 2 & 4 & 6 & 8 & 10 \\ \hline
MP       & $\approx10^{-8}$  & $\approx10^{-6}$  &  $\approx10^{-5}$ & N/A  &  N/A  \\
TISP-PG  & 0.053  & 0.112  & 0.211  & 0.329  & 0.473   \\
TISP-CFR & 0.008 & 0.065  & 0.190  &  0.331 &  0.499
\end{tabular}
\caption{$|A|=2$}
\end{subtable}
\begin{subtable}[h]{0.45\textwidth}
\small
\centering
\begin{tabular}{l|lllll}
L        & 2 & 4 & 6 & 8 & 10 \\ \hline
MP       &  $\approx10^{-6}$ &  $\approx10^{-6}$ & $\approx10^{-3}$  & N/A  &  N/A  \\
TISP-PG  &  0.120 & 0.232  & 0.408  & 0.599  & 0.842   \\
TISP-CFR &  0.002 &  0.049 &  0.285 & 0.525  &  0.847 
\end{tabular}
\caption{$|A|=5$}
\end{subtable}

\caption{Comparing mathematical-programming 
and our methods, i.e., TISP-PG and TISP-CFR, with known model. These results are averaged over 21 starting prior distributions of the attacker ($[0.00, 1.00], [0.05, 0.95],\dots, [1.00, 0.00]$).}
\label{tab:sec_game_math}
\end{table}


\subsubsection{Results}
We first experiment with the zero-sum setting where we have proved our model can converge to an $\epsilon$-PBE.
The comparison are shown in Table~\ref{table:zero_sum_sec_game},\ref{table:known_model_sec_game}. We use two known model variants, TISP-PG and TISP-CFR, and a model-free version of TISP-PG in this comparison. TISP-PG achieves the best results, while TISP-CFR also has comparable performances. 
We note that simply using an RNN or using belief-space policy performs only slightly better than a random policy. 

Then we conduct experiments in general-sum games with results shown in Table~\ref{table:general_sum_sec_game},\ref{table:known_model_sec_game}. We empirically observe that the derived solution has comparable quality with the zero-sum setting. We also compare our methods with the Mathematical-Programming-based method (MP) in \cite{nguyen2019deception}, which requires full access to the game transition. The results are shown in Table~\ref{tab:sec_game_math}. Although when $L$ is small, the MP solution achieves superior accuracy, it quickly runs out of memory (marked as ``N/A'') since its time and memory requirement grows at an exponential rate w.r.t. $L$. Again, our TISP variants perform the best among all learning-based methods. We remark that despite of the performance gap between our approach and the MP methods, the error on those games unsolvable for MP methods is merely $~0.1\%$ comparing to the total utility.

In our experiments, we do not observe orders of magnitudes difference in running time between our method and baselines: TISP-PG and TISP-CFR in the Tagging game uses 20 hours with 10M samples in total even with particle-based approximation (200k samples per belief point) while RNN and BPG in the Tagging game utilizes roughly 7 hours with 2M total samples for convergence. 

Regarding the scalability on the number of types, as the number of types increases, the intrinsic learning difficulty increases. This is a challenge faced by all the methods. The primary contribution of this paper is a new learning-based framework for PBNE. While we primarily focus on the case of 2 types, our approach generalizes to more types naturally with an additional experiment conducted for 3 types in Table.~\ref{table:3types_sec_game}. Advanced sampling techniques can potentially be utilized for more types, which we leave for future work. 

\subsection{Exposing Game}

\subsubsection{Game Setting} We also present a two-step matrix game, which we call \textit{Exposing}. In this game, player 2 aims to guess the correct type of player 1 in the second round. There are two actions available for player 1 and three actions available for player 2 in each round. Specifically, the three actions for player 2 means guessing player 1 is \textit{type 1}, \textit{type 2} or not guessing at all. The reward for a correct and wrong guess is $10$ and $-20$ respectively.  The reward for not guessing is $0$. Player 1 receives a positive $5$ reward when the player 2 chooses to guess in the second round, regardless of which type player 2 guesses. In this game, player 1 has the incentive to expose its type to encourage player 2 to guess in the second round. We further add a reward of $1$ for player 1 choosing action 1 in the first round regardless of its type to give player 1 an incentive to not exposing its type. With this reward, a short-sighted player 1 may pursue the $1$ reward and forgo the $5$ reward in the second round. The payoff matrices for this game are in Appx.~\ref{app-sec:expr}.

The equilibrium strategy for player 2 in the second round w.r.t. different types of player 1 is:

\begin{equation*}
    \text{strategy}:=\left\{ 
    \begin{array}{cc}
        \text{guessing type 1} & \text{if }\Pr[1|h]>\frac{1}{3} \\
        \text{guessing type 2} & \text{if }\Pr[2|h]>\frac{1}{3} \\
        \text{not guessing} & \text{else}
    \end{array}
    \right.
\end{equation*}

\begin{figure}[t]
    \centering
    \subfloat[Ground truth]{\includegraphics[width=.14\textwidth]{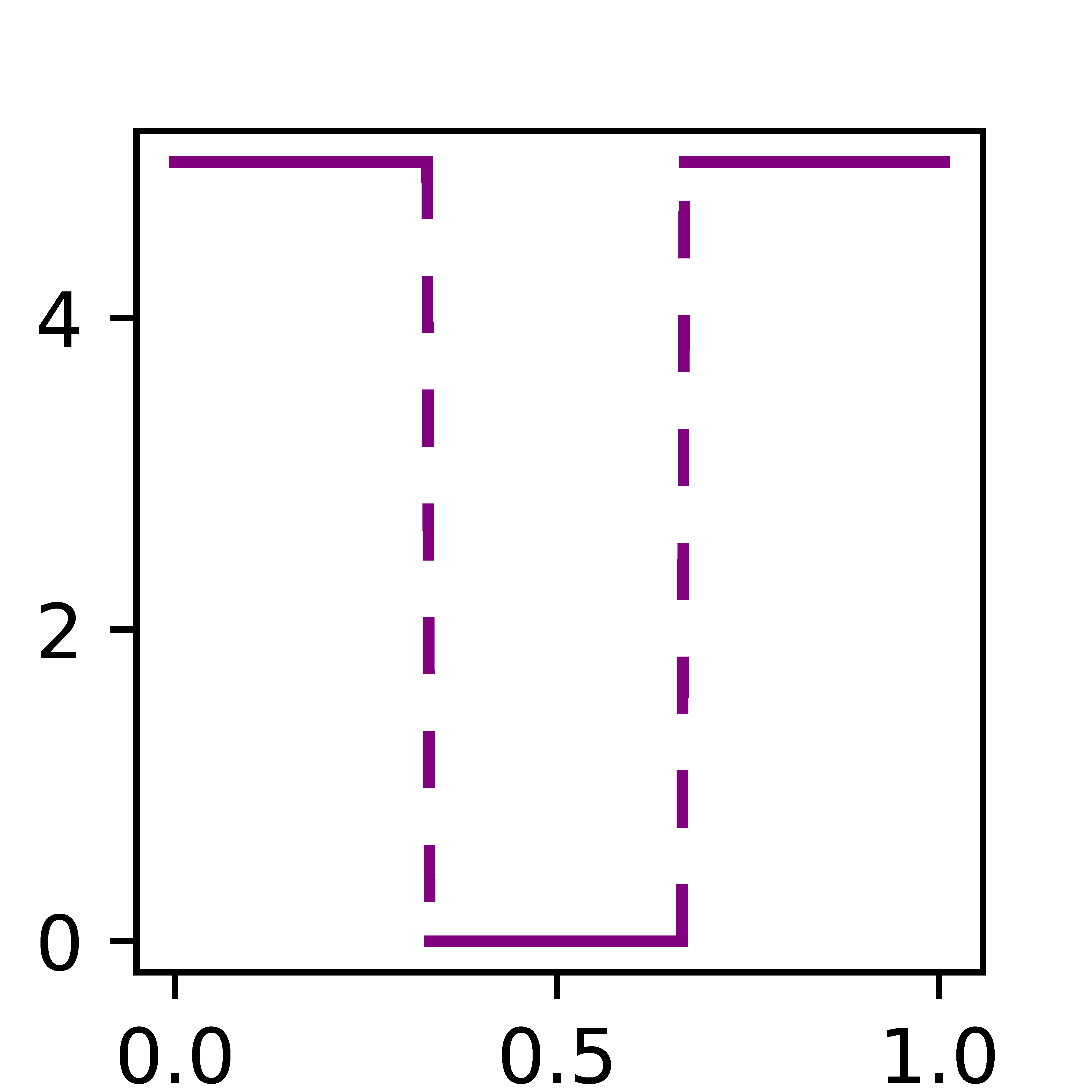} \label{fig:exposing-truth}}
    \qquad
    \subfloat[Approximation]{\includegraphics[width=.14\textwidth]{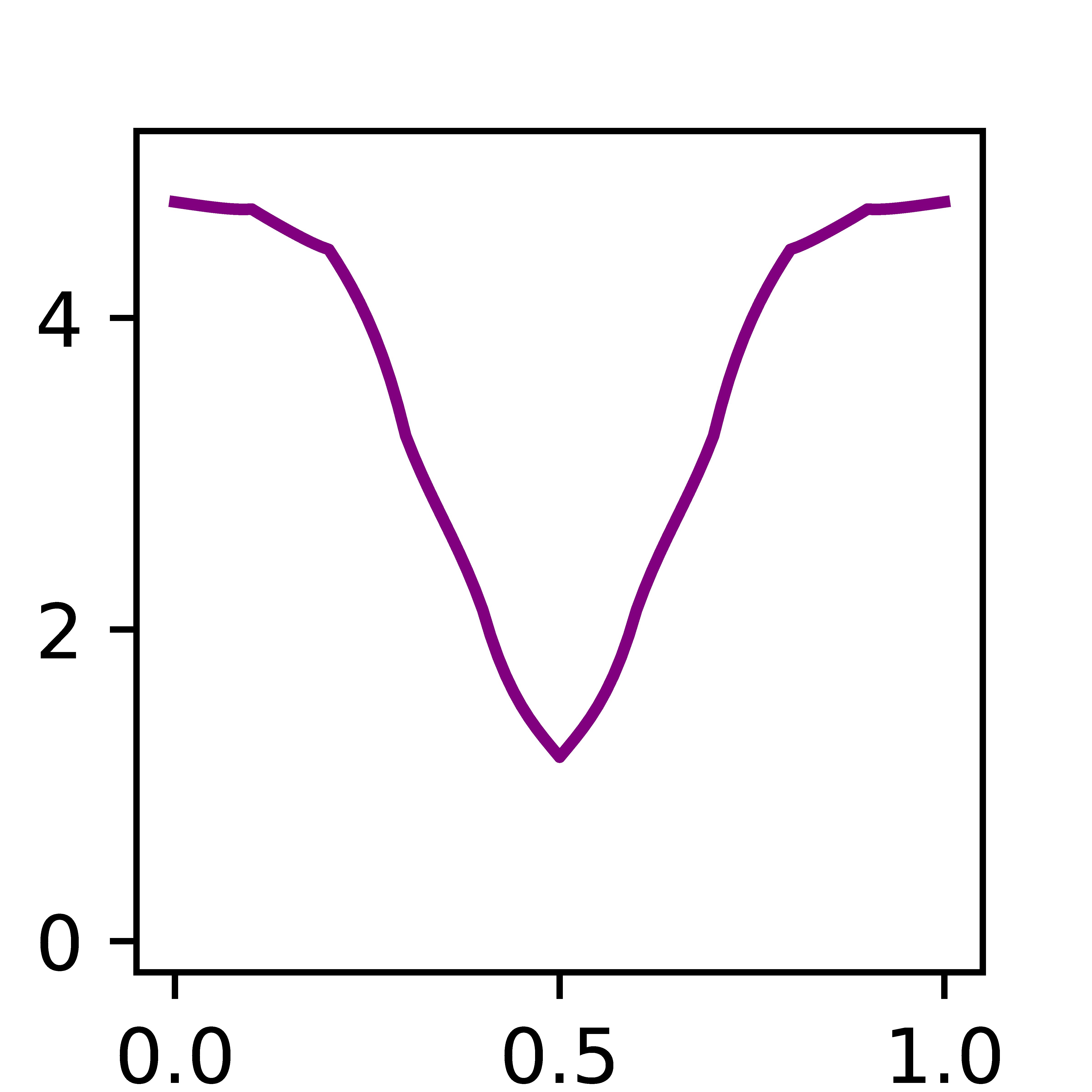} \label{fig:exposing-approx}}
    \caption{Ground truth and approximated value of player 1 in the second round of Exposing game. The $x-$axis corresponds to $\Pr[1|h]$. The $y-$axis corresponds to the equilibrium value. }
    \label{fig:exposing-value}
\end{figure}

In this game, the equilibrium values in the second round for both types of player 1 are highly discontinuous with regard to player 2's belief, as shown in Fig.~\ref{fig:exposing-truth}, which makes the belief-space gradient term in Eq.~\ref{eq:pg_update} ineffective. However, the approximation introduced in Sec.~\ref{sec:approx} can serve as a soft representation of the true value function, as shown in Fig.~\ref{fig:exposing-approx}. This soft representation provides an approximated gradient in the belief space, which allows for belief-space exploration. We will show 
later that this belief-space exploration cannot be achieved otherwise.


\begin{table}[t]
\small
\begin{tabular}{l|l|lll}
 & P1's type      & Action 1 & Action 2 & Reward \\ \hline
\multirow{2}{*}{TISP-PG}      & type 1  &   0.985& 0.015 & 5.985\\  \cline{2-5} 
   & type 2 & 0.258& 0.742 & 5.258 \\ \hline
\multirow{2}{*}{TISP-CFR}       & type 1  &  1.000& 0.000 & 1.000  \\ \cline{2-5}
    & type 2 & 1.000& 0.000 & 1.000  \\ \hline
 \multirow{2}{*}{TISP-PG$^-$}       & type 1  & 0.969 &  0.031 & 0.969  \\ \cline{2-5} 
  & type 2 & 0.969  & 0.031 & 0.969   \\ \hline
 \multirow{2}{*}{Optimal}       & type 1  & 1.000 &  0.000 & 6.000 \\ \cline{2-5} 
  & type 2 & 0.333  & 0.667 & 5.333
\end{tabular}
\caption{Detailed first round policy in Exposing game. TISP-PG is the only algorithm that separates the two player-1 types' strategies and yields a result very close to the optimal solution.\protect\footnotemark}
\label{table:exposing-policy}
\end{table}

\footnotetext{The optimal solution refers to the Pareto-optimal PBE in this game. Note that there are two symmetric optimal solutions. We choose to only show one here for easy comparison and simplicity.}

\subsubsection{Results} We compare the training result between TISP-PG and TISP-CFR to exhibit the effectiveness of our non-parametric approximation. We further add an ablation study that removes the belief-space gradient term in TISP-PG, which we call TISP-PG$^-$. The results are shown in Table~\ref{table:exposing-policy}. We can see that TISP-PG is the only algorithm that successfully escapes from the basin area in Fig.~\ref{fig:exposing-value} as it is the only algorithm that is capable of doing belief-space exploration. We also show the obtained policies  Table~\ref{table:exposing-policy}. The training curves can be found in Appx.~\ref{app-sec:expr}. Note that the policies trained from TISP-CFR and TISP-PG$^-$ are also close to a PBE where player 1 takes action 1 regardless of its type in the first round, and player 2 chooses not to guess in the second round, although it is not Pareto-optimal.


\subsection{Tagging Game}
\begin{figure}[t]
    \centering
    \includegraphics[width=0.8\linewidth]{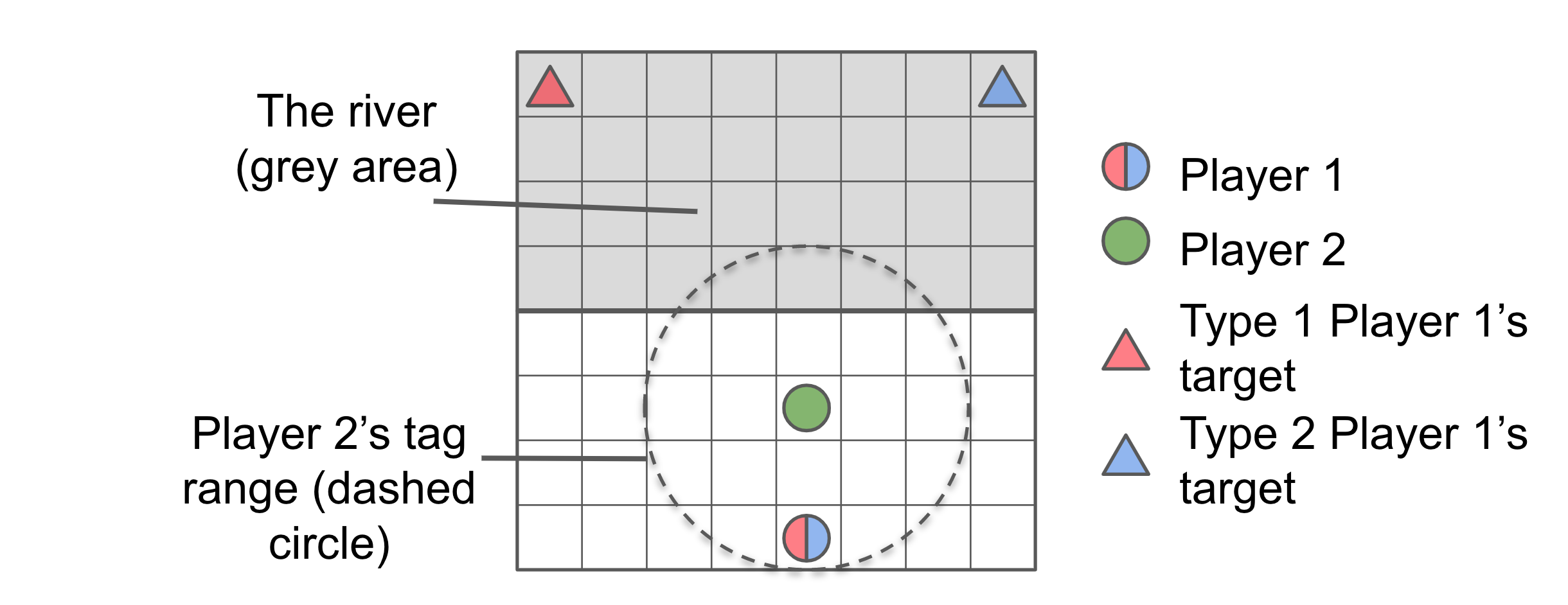}
    \caption{An illustration of the Tagging game.}
    \label{fig:tagging}
\end{figure}
\subsubsection{Game Setting} We test our method in a gridworld game \emph{Tagging}, as illustrated in Fig.~\ref{fig:tagging}. The game is inspired by  \cite{shen2019robust}. Specifically, the game is on an $8 \times 8$ square, and player 1 has two types, i.e., \textit{ally} and \textit{enemy}, and each type corresponds to a  unique target place. Each player will receive distance-based reward  to encourage it to move towards its target place. 
There is a river in the top half part of the grids which Player 2 cannot enter. Player 1 starts from the bottom middle of the map and player 2 starts from a random position under the river. Both players can choose to move in one of the four directions, [up(U), down(D), left(L), right(R)], by one cell. Player 2 has an additional action, \textit{tag}, to tag player 1 as the \textit{enemy} type. The tag action is only available when player 1 has not entered the river and the euclidean distance between the two players is less than $2.5$. The attackers get higher rewards for getting closer to their type-specified target, and the defenders get higher rewards for getting closer to the attacker. Moreover, if the defender chooses to tag, it will get a reward of $10$ if the attacker is of the \textit{enemy} type and a reward of $-20$ if the attacker is of the \textit{ally} type, while the attacker will get a reward of $-10$ for being tagged no matter what its type is. More detail is provided in Appx.~\ref{app-sec:expr}.

Based on the game's rule, an enemy-typed player 1 is likely to go to its own target immediately to get a higher reward. 
However, such a strategy reveals its type straight away, which can be punished by being tagged. A more clever strategy of an enemy-typed player 1 is to mimic the behavior of the ally-typed player 1 in most situations and never let player-2's belief of enemy type  be larger than $\frac{2}{3}$, so that player 2 has no incentive to tag it. The ally-typed player 1 may simply go up in most states in order to get closer to its target for a reward while player 2 should learn to tag player 1 if 
its belief of  enemy type is high enough and try to move closer to player 1 in other cases. 

\subsubsection{Evaluation}
Although it is computationally intractable to calculate the exact exploitability in this gridworld game, we examine the results following \cite{gray2020human} by evaluating the performances in induced games. We choose 256 induced games among all the induced games after the second round and check their exploitability. We get the best response in each induced game in two ways: in the first round of each game, we enumerate all the possible actions and manually choose the best action. We also train a new BPG player-2 from scratch. This new BPG-based player is trained in a single-agent-like environment and does not need to consider the change in the player-1 policy. We check how much reward agents can get after the learning process finishes. A high reward from player 2 would indicate that player 1 fail to hide its type.  We also provide the detailed strategies from different baselines for the first round of the game for additional insight.

\begin{table}[t]
\small
\begin{tabular}{l|l|llll}
 & P1's Type      & U & D & R & L  \\ \hline
TISP-PG      & Ally  &   0.839& 0.001& 0.001& 0.159    \\ \cline{2-6} 
  &  Enemy & 0.932& 0.001& 0.001& 0.066      \\ \hline
RNN & Ally  &  0.248& 0.237& 0.238& 0.274       \\ \cline{2-6}
    & Enemy & 0.599& 0.067& 0.077& 0.255       \\ \hline
BPG    & Ally  & 0.000  &  0.000 & 0.000  & 1.000      \\ \cline{2-6} 
    & Enemy & 1.000  & 0.000  & 0.000  & 0.000   \\ \hline
TISP-CFR & Ally  & 0.000  & 0.000  & 0.000  &  1.000    \\ \cline{2-6} 
  & Enemy & 1.000  & 0.000   & 0.000 & 0.000    \\ 
\end{tabular}
\caption{The policy at one of the starting states in Tagging game, where player 2 is two cells above the fixed starting point of player 1. 
}
\label{table:tagging_policy}
\end{table}


\begin{table}[t]
\small
\centering
\begin{tabular}{l|llll}
    & TISP-PG & RNN & BPG & TISP-CFR\\ \hline
P2 reward & -1.90 &   -1.67   &  -0.98& -1.82\\
P1 reward (ally) & -2.55 & -2.87 & -3.26&-3.17\\
P1 reward (enemy) & -2.41 & -2.71 & -9.29&-4.49\\
\end{tabular}
\caption{The average exploitability result of 256 induced games in the Tagging game. The lower player 2's reward, the better the algorithm.}
\label{table:tagging_exploitability}
\end{table}

\subsubsection{Results}
We show the derived policy in the very first step in Table~\ref{table:tagging_policy}. The policy learned by our method successfully  keeps the belief to be less than $\frac{1}{3}$, and keeps a large probability of going to the target of each type. The RNN policy shows no preference between different actions, resulting in not being tagged but also not getting a larger reward for getting closer to the target. The BPG policy simply goes straight towards the target and is therefore punished for being tagged. 
The exploitability results are shown in Table~\ref{table:tagging_exploitability}. 
From the training reward achieved by the new exploiter player 2, TISP-PG performs the best among all baselines and TISP-CFR also produces a robust player-1 policy. We remark that relative performances of different methods in the gridworld game are consistent with what we have previously observed in the finitely repeated security games, which further validates the effectiveness of our approach. 

%% file: 07Discussion.tex
\section{Discussion and Conclusion}
\label{sec:conclusion}
We proposed TISP, an RL-based framework to find strategies with a decent performance from any decision point onward. We provided theoretical justification and empirically demonstrated its effectiveness. Our algorithms can be easily extended to a two-sided stochastic Bayesian game. The TISP framework still applies, and the only major modification needed is to add a for-loop to sample belief points for player 2. This extension will cost more computing resources, but the networks can be trained fully in parallel. The full version of the extended algorithm is in Appendix~\ref{app-sec:extend}.

%% file: appendix.tex
\appendix

\section{Discussion on Belief Sampling}\label{app-sec:belief-sampling}
Our algorithm generally requires the sampled belief points guarantee that any possible belief has a sampled one less than $d$ distance from it, where $d$ is a hyperparameter. In one-dimension belief space (for two types), the most efficient way is to draw these beliefs with equal distance. 
This leads to a result that assures not be too far from any point that may cause policy change. With more sampled beliefs, the belief-space approximation will be more precise, which also requires more training resources accordingly.
An alternative solution is adaptive sampling based on the quality of the obtained policies, which,  however, requires sequential execution at different belief points at the cost of restraining the level of parallelism. The complete training framework is summarized in Algo.~\ref{app-alg:ti_framework}.

\section{Implementation Details}\label{app-sec:impl}
\subsubsection*{Sampling New Games in Backward Induction}
Implementation-wise, we assume access to an auxiliary function from the environment, which we called $\textit{sub\_reset}(l,b)$. This function takes two parameters, a round $l$ and a belief $b$, as input and produces a new game by drawing a random state from the entire state space $\Omega$ with equal probability and a random type according to the belief distribution $b$.
This function is an algorithmic requirement for the environment, which is typically feasible in practice. For example, most RL environments provides a \textit{reset} function that generates a random starting state, so a simple code-level enhancement on this \textit{reset} function can make existing testbeds compatible with our algorithm.
We remark that even with such a minimal environment enhancement requirement, our framework does \emph{NOT} utilize the transition information. Hence, our method remains nearly model-free comparing to other methods that assume full access to the underlying environment transitions --- this is the assumption of most CFR-based algorithms. 

\subsubsection*{RNN}
We use Gated Recurrent Unit networks~\cite{Cho2014LearningPR} to encode the state-action history and perform a top-down level self-play.

\subsubsection*{TISP-PG}

\begin{algorithm}[h]
\caption{TISP-PG}
\label{app-alg:tipg}

\begin{algorithmic}[1]
\Function{Training}{}
\For {$l=L, L - 1, \dots, 0$}
    \For {$k=1, 2, \dots K$} \Comment{ This loop can run in fully parallel.}
        
        \State Initialize the supervise set $D=\{\}$
        \For {$t=1, 2, \dots, T$} 
            \State $states \gets sub\_reset(l, b_k)$;
            
            \State $acts \gets Sampled\ base\ on\ \theta_{1,l, k}, \theta_{2, l, k}, states$;
            
            \State $s', rews, done \gets env.step(acts)$;
            
            \If {not done} 
                $rewards \gets rews + \gamma V_{\phi_{l+1}}(s')$;
            \EndIf
            
            \State $\theta_{1, l, k}, \theta_{2, l, k} \gets $ PGUpdate(states, acts, rews);
        \EndFor
        
        \State Initialize the supervise set $D=\{\}$
        
        \For {$t=1, 2, \dots, T$}
            \State $states \gets sub\_reset(l, b_k)$;
            \State $acts \gets Sampled \ base \  on \  \theta_{1, l, k}, \theta_{2, l, k}, states$;
            \State$s', rews, done \gets env.step(acts)$;
            
            \If {not done} 
                \State $rews \gets rews + \gamma V_{\phi_{l+1}}(s')$;
            \EndIf
            \State $D \leftarrow D + (states, rews)$;
        \EndFor
        \State $\phi_{1, l, k}, \phi_{2, l, k} \gets SuperviseTraining(D)$;
    \EndFor
\EndFor
\State \Return All $L \times K$ groups $\theta_1,  \theta_2, \phi_1, \phi_2$;
\EndFunction
\end{algorithmic}
\end{algorithm}

We implement our method TISP-PG as shown above. Specifically, there are two ways to implement the attacker, for we can either use separate networks for each type or use one network for all types and take the type simply as additional dimension in the input. In this paper, We use separate network for each type.





\subsubsection*{TISP-CFR}

The pseudo-code of our TISP-CFR is shown here:

\begin{algorithm}[h]
\caption{TISP-CFR}
\label{app-alg:ticfr}

\begin{algorithmic}[1]
\Function{Training}{}
\For {$l=L, L - 1, \dots, 0$}
    \For {$k=1, 2, \dots K$} \Comment{ This loop can run in fully parallel.}
        
        \State Initialize the supervise set $D=\{\}$
        \For {$t=1, 2, \dots, T$}
            \State $s_1, s_2 \gets sub\_reset(l, b_k)$;
            \State $acts \sim  (\pi_{1, l, k}(s_1), \pi_{2, l, k}(s_2) )$;
            
            \State $s', rewards, done \gets env.step(acts)$;
            
            \If {not done}
                \State $rewards \gets rewards + \gamma V_{\phi_{l+1}}(s')$;
            \EndIf
            \State $D \gets D + (states, rewards)$;
        \EndFor
        \State Update ${V_\phi}_{1, l, k}, {V_\phi}_{2, l, k}$  using D;
        \State Calculate $\pi_{1, l, k}, \pi_{2, l, k}$ using regret matching;
    \EndFor
\EndFor
\State \Return $All\ L \times K \ groups\ V \ and\ \pi$;
\EndFunction
\end{algorithmic}
\end{algorithm}

\section{Extension algorithm to two-sided Bayesian game}\label{app-sec:extend}

\begin{algorithm}[h]
\caption{Temporal-induced Self-Play for two-sided Stochastic Bayesian games}
\label{app-alg:ti_framework}

\begin{algorithmic}[1]
\Function{Training}{}
\For {$l=L, L - 1, \dots, 0$}
    \For {$k_1=1, 2, \dots K_1$} 
        \For {$k_2=1, 2, \dots K_2$} 
    
        \Comment{ This loop can run in fully parallel.}
        
        \State Initialize the supervise set $D=\{\}$
        \For {$t=1, 2, \dots, T$}
            \State $s_1, s_2 \gets sub\_reset(l, b_{k_1}, b_{k_2})$;
            \State $acts \sim  (\pi_{1, l, k_1,k_2}(s_1), \pi_{2, l, k_1,k_2}(s_2) )$;
            
            \State $s', rews, done \gets env.step(acts)$;
            
            \If {not done}
                \State $rews \gets rews + \gamma V_{\phi_{l+1}}(s')$;
            \EndIf
            \State $D \gets D + (states, rews)$;
        \EndFor
        \State Update ${V_\phi}_{1, l, k_1,k_2}, {V_\phi}_{2, l, k_1,k_2}$  using D;
        \State Update $\pi_{1, l, k_1,k_2}, \pi_{2, l, k_1,k_2}$; 
        \EndFor
    \EndFor
\EndFor
\State \Return $All\ L \times K \ groups\ V \ and\ \pi$;
\EndFunction
\end{algorithmic}
\end{algorithm}

Our algorithm can be extended to two-sided Stochastic Bayesian game, with the learning algorithm shown above. Specifically, we now do belief approximation in both side of players and calculate all the policies. This means all the equations in the paper, e.g., Eq. 4, 5, the belief vector $b$ should now be $b_1, b_2$ which are the belief points in both side. In other word, if you see the belief vector $b$ to carry the belief on both the type of both players, the equations can stay with the same.

\section{Experiment Details}~\label{app-sec:expr}

\subsection{Computing Infrastructure}

All the experiments on Security game and Exposing game are conducted on a laptop with 12 core Intel(R) Core(TM) i7-9750H CPU @ 2.60GHz, 16 Gigabytes of RAM and Ubuntu 20.04.1 LTS. All the experiments on Tagging game are conducted on a server with two 64 core AMD EPYC 7742 64-Core Processor @ 3.40GHz, 64 Gigabytes of RAM and Linux for Server 4.15.0.

\begin{table}[t]
    \centering
    \caption{Comparing mathematical-programming 
and our methods, i.e., TISP-PG and TISP-CFR, with known model and $|A| =2,|\Lambda|=2$. These results are averaged over 21 starting prior distributions of the attacker ($[0.00, 1.00], [0.05, 0.95],\dots, [1.00, 0.00]$).}
    \begin{tabular}{c|ccc}
$L$ & MP & TISP-PG & TISP-CFR  \\
\hline
1 & $<10^{-8}$ & 0.034 & 0.003\\
2 & $\approx 10^{-8}$ &  0.053 &0.008 \\
3 & $\approx 10^{-7}$ & 0.083 &0.030 \\
4 & $\approx 10^{-6}$ & 0.112 &0.065 \\
5 & $\approx 10^{-5}$ & 0.162 &0.117 \\
6 & $\approx 10^{-5}$ & 0.211 & 0.190\\
7 & N/A & 0.267& 0.251\\
8 & N/A & 0.329&0.331 \\
 9 & N/A & 0.448&0.459 \\
10 & N/A & 0.473& 0.499
\end{tabular}
    
    \label{tab:sec_game_all}
\end{table}

\subsection{Security Game}

We have also conducted the experiments in an easier setting of $|A|=2$, the result is shown in Table.~\ref{tab:sec_game_all}.

\subsection{Exposing Game}

\begin{figure}[t]
    \centering
    \captionsetup[subfigure]{labelformat=empty}
    \subfloat[Round 1]{
    \captionsetup[subfigure]{labelformat=empty}
    \subfloat[Type 1]{
        \begin{tabular}{|c|c|c|}
            \hline
            $(1, 0)$ & $(1, 0)$ & $(1, 0)$ \\
            \hline
            $(0, 0)$ & $(0, 0)$ & $(0, 0)$ \\
            \hline
        \end{tabular}
    }
    \qquad
    \subfloat[Type 2]{
        \begin{tabular}{|c|c|c|}
            \hline
            $(1, 0)$ & $(1, 0)$ & $(1, 0)$ \\
            \hline
            $(0, 0)$ & $(0, 0)$ & $(0, 0)$ \\
            \hline
        \end{tabular}
    }
    }
    \\
    \resizebox{0.99\columnwidth}{!}{
    \subfloat[Round 2]{
    \captionsetup[subfigure]{labelformat=empty}
    \subfloat[Type 1]{
        \begin{tabular}{|c|c|c|}
            \hline
            $(5, 10)$ & $(5, -20)$ & $(0, 0)$ \\
            \hline
            $(5, 10)$ & $(5, -20)$ & $(0, 0)$ \\
            \hline
        \end{tabular}
    }
    \qquad
    \subfloat[Type 2]{
        \begin{tabular}{|c|c|c|}
            \hline
            $(5, -20)$ & $(5, 10)$ & $(0, 0)$ \\
            \hline
            $(5, -20)$ & $(5, 10)$ & $(0, 0)$ \\
            \hline
        \end{tabular}
    }
    }
    }
    \caption{Payoff matrices for Exposing game. The first number in the tuple indicates the payoff for player 1.}
    \label{fig:expoising-matrices}
\end{figure}

The payoff matrices are shown in Fig.~\ref{fig:expoising-matrices}

The training curves for the three algorithms are shown in Fig.~\ref{fig:exposing-curve}. In the latter two plots, the strategies for both types of player 1 stay the same throughout the training process. This is why there is only one curve visible.

\begin{figure}[t]
    \centering
    \includegraphics[width=.8\linewidth]{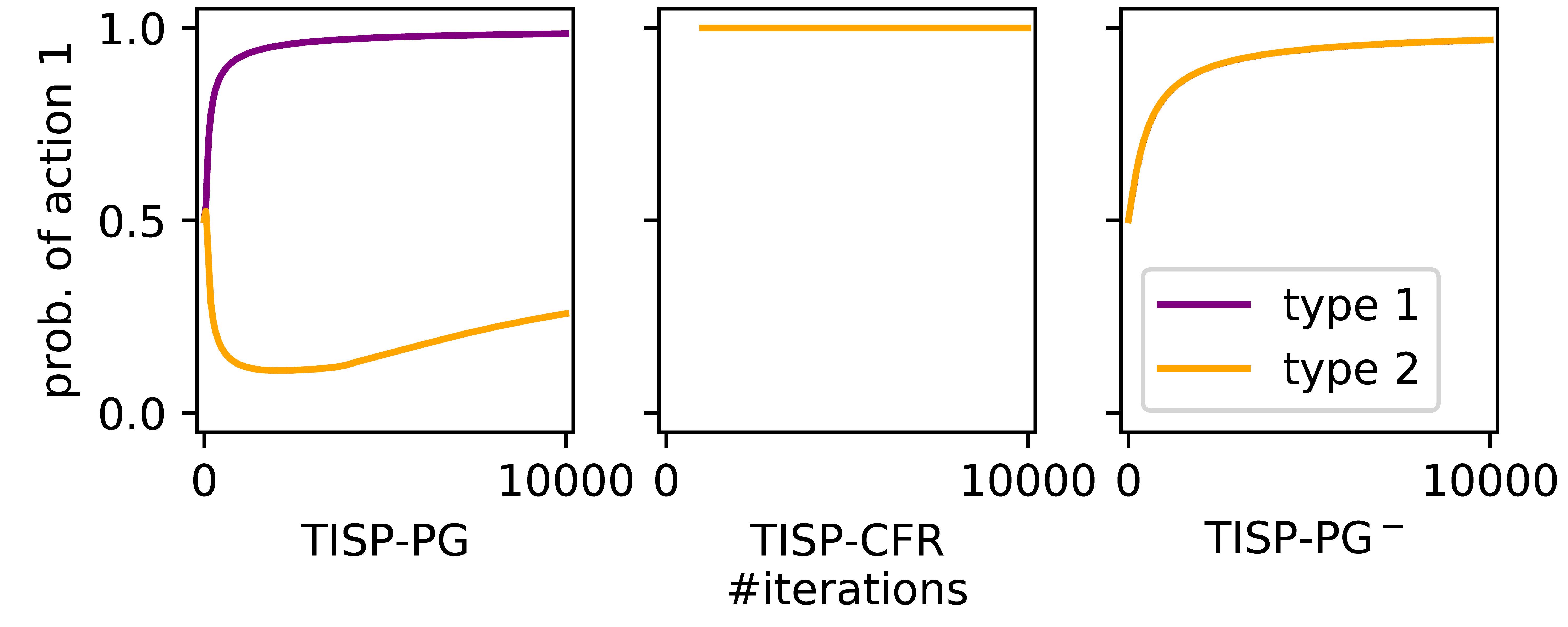}
    \caption{Training curves for the first round of Exposing.}
    \label{fig:exposing-curve}
\end{figure}

\subsection{Tagging game}

The target for ally is on the upper left corner, and the target for enemy is on the upper right corner. The rewards for both players can be written in forms of :
    $r = r_{d} + r_{tag}$
, where $r_d$ is the reward for players to get close to its own targets, and $r_{tag}$ is the reward that is non-zero only when player 2 takes tag action. Specifically for player 1, $r_{d} = -0.25Dist(a, target)^{\frac{2}{5}}$, where $Dist(1, target)$ means the euclidean distance between player 1 and the target of its type. $r_{tag} = -10$ if player 2 tags player 1 in that round and $r_{tag}=0$ if player 2 does not do tag action in this round. Note that even if player 1 is actually an \textit{ally} type, it still gets a penalty. 
For player 2, $r_{d} = -0.25Dist(1, 2)^{\frac{2}{5}}$, where $Dist(1, 2)$ means the euclidean distance between the two players. $r_{tag} = -10$ if player 1 is of ally type and is tagged by player 2 this round, $r_{tag} = 20$ if player 1 is of enemy type and is tagged by player 2 this round. The prior distribution of ally type and enemy type is $[0.5, 0.5]$. The tag action can be used multiple times through the game. We set an episode length limit of $5$ in our experiments. While $5$ steps seem not able to make both player reach their target, the environment does not give any huge reward for reaching the target.
Remark that We assume a player can only know its reward when the game ends, preventing player 2 from knowing player 1's type by immediately looking at its reward after taking a tag action.

\section{Changes to Policy Gradient Corresponding to Belief-space Policy}~\label{app-sec:pg}

\begin{proof}

We demonstrate the proof for player 1. The proof for player 2 should be basically the same with some minor modifications. 

\small{
\begin{align}
    \nabla_{\theta_{1, \lambda}} v^{\pi_1,\pi_2}_{1,\lambda}(b, s)=\nabla_{\theta_{1,  \lambda}}\sum_{a_1}\pi_{1, \lambda}(a_1|b, s;\theta_{1, \lambda})
    \sum_{a_2}\pi_{2}(a_2|b, s;\theta_2)\nonumber\\
    \bigg(u_{1,\lambda}(s, a_1, a_2)
    +\gamma \sum_{s'}P(s'|s, a_1, a_2)v^{\pi_1, 
    \pi_2}_{1, \lambda}(b', s')\bigg)
\end{align}
}

Let $Q^{\pi_1, \pi_2}_{1,\lambda}(b, s, a_1, a_2)=u_{1,\lambda}(s, a_1, a_2)+\gamma \sum_{s'}P(s'|s, a_1, a_2)v^{\pi_1, \pi_2}_{1,\lambda}(b, s')$. Then we have

\begin{align}
    \nabla_{\theta_{1, \lambda}} v^{\pi_1,\pi_2}_{1, \lambda}(b, s)=&\nabla_{\theta_{1, \lambda}}\sum_{a_1}\pi_{1, \lambda}(a_1|b, s;\theta_{1, \lambda})\sum_{a_2}\pi_{2}(a_2|b, s;\theta_2)\nonumber\\
    &\cdot Q^{\pi_1, \pi_2}_{1,\lambda}(b', s, a_1, a_2)\nonumber\\
    =&\sum_{a_2}\pi_2(a_2|b, s;\theta_2)\sum_{a_1}\bigg(Q^{\pi_1, \pi_2}_{1,\lambda}(b', s, a_1, a_2)\nonumber\\
    &\cdot \nabla_{\theta_{1,\lambda}}\pi_{1,\lambda}(a_1|b, s;\theta_{1,\lambda})+\pi_{1,\lambda}(a_1|b, s;\theta_{1,\lambda})\nonumber\\
    &\qquad\cdot\nabla_{\theta_{1,\lambda}}Q^{\pi_1, \pi_2}_{1,\lambda}(b', s, a_1, a_2)\bigg)
    \label{eq:v}
\end{align}

where $b'=\text{Bayes}(b, \pi_1)$. Now we expand the term $\nabla_{\theta_{1,\lambda}}Q^{\pi_1, \pi_2}_{1,\lambda}(b', s, a_1, a_2)$:

\begin{align}
    \nabla_{\theta_{1,\lambda}}&Q^{\pi_1, \pi_2}_{1,\lambda}(\text{Bayes}(b, \pi_1), s, a_1, a_2)\\
    =& \nabla_{\theta_{1,\lambda}}\text{Bayes}(b, \pi_1)\nabla_{b'}Q^{\pi_1, \pi_2}_{1,\lambda}(b', s, a_1, a_2)\nonumber\\
    &\qquad +\nabla_{\theta_{1,\lambda}}Q^{\pi_1, \pi_2}_{1,\lambda}(b', s, a_1, a_2)\nonumber\\
    =& \nabla_{\theta_{1,\lambda}}\text{Bayes}(b, \pi_1)\nabla_{b'}Q^{\pi_1, \pi_2}_{1,\lambda}(b', s, a_1, a_2) \nonumber \\
    &\qquad+ \gamma \sum_{s'}P(s'|s, a_1, a_2)\nabla_{\theta_{1,\lambda}}v^{\pi_1,\pi_2}_{1,\lambda}(b', s')
    \label{eq:q}
\end{align}

Let $\phi(s)=\sum_{a_2}\pi_2(a_2|b, s;\theta_2)\sum_{a_1}\big(Q^{\pi_1, \pi_2}_{1,\lambda}(b', s, a_1, a_2)$\\$\nabla_{\theta_{1,\lambda}}\pi_{1,\lambda}(a_1|b, s;\theta_{1,\lambda})+\gamma \pi_{1,\lambda}(a_1|b, s;\theta_{1,\lambda})\nabla_{\theta_{1,\lambda}}\text{Bayes}(b, \pi_1)$\\$\nabla_{b'}Q^{\pi_1, \pi_2}_{1,\lambda}(b', s, a_1, a_2)\big)$. Together with (\ref{eq:q}), we can rewrite (\ref{eq:v}) to:

\small{
\begin{align}
    \nabla_{\theta_{1, \lambda}} v^{\pi_1,\pi_2}_{1, \lambda}(b, s)=&\phi(s)+\gamma \sum_{a_2}\pi_2(a_2|b, s;\theta_2)\sum_{a_1}\pi_{1,\lambda}(a_1|b, s;\theta_{1,\lambda})\nonumber\\
    &\qquad\cdot\sum_{s'}P(s'|s, a_1, a_2)\nabla_{\theta_{1,\lambda}}v^{\pi_1,\pi_2}_{1,\lambda}(b', s')
\end{align}}

With policy gradient theorem, we have:

\begin{align}
    \nabla_{\theta_{1, t}}& J(\theta_{1,\lambda})\propto \sum_{s,b}\sum_{a_2}\pi_2(a_2|b, s;\theta_2)\sum_{a_1}\bigg(Q^{\pi_1, \pi_2}_{1,\lambda}(b', s, a_1, a_2)\nonumber\\
    &\qquad\nabla_{\theta_{1,\lambda}}\pi_{1,\lambda}(a_1|b, s;\theta_{1,\lambda})+\gamma\pi_{1,\lambda}(a_1|b, s;\theta_{1,\lambda})\nonumber\\
    &\qquad\nabla_{\theta_{1,\lambda}}\text{Bayes}(b, \pi_1)\nabla_{b'}Q^{\pi_1, \pi_2}_{1,\lambda}(b', s, a_1, a_2)\bigg)\nonumber\\
    =&\mathbb{E}_{s,b,a_1, a_2}\bigg[Q^{\pi_1, \pi_2}_{1,\lambda}(b', s, a_1, a_2)\nabla_{\theta_{1,\lambda}}\ln\pi_{1,\lambda}(a_1|b, s;\theta_{1,\lambda})\nonumber\\
    &\qquad+\gamma\nabla_{\theta_{1,\lambda}}\text{Bayes}(b, \pi_1)\nabla_{b'}Q^{\pi_1, \pi_2}_{1,\lambda}(b', s, a_1, a_2)\bigg]\nonumber\\
    =&\mathbb{E}_{s,b,a_1, a_2}\bigg[Q^{\pi_1, \pi_2}_{1,\lambda}(b', s, a_1, a_2)\nabla_{\theta_{1,\lambda}}\ln\pi_{1,\lambda}(a_1|b, s;\theta_{1,\lambda})\nonumber\\
    &\qquad+\gamma\nabla_{\theta_{1,\lambda}}\text{Bayes}(b, \pi_1)\nabla_{b'}v^{\pi_1, \pi_2}_{1,\lambda}(b', s')\bigg]
\end{align}

\end{proof}

\section{Proof of the Main Theorem}\label{app-sec:main-proof}

In this proof, we adopt a slightly different notation system for better interpretability.

\begin{itemize}
	\item We use $\sigma_i$ instead of $\pi_i$ to denote Player $i$'s strategy; $\pi$ in this proof will be used to denote the reaching probability in the game tree.
	\item We use $\bar{u}_i^\lambda(s,a_1,a_2)$ to denote the utility function in the original game instead of $u_i^\lambda(s,a_1,a_2)$; with some abuse of notation, we use $u_i$ generally to denote the expected utility of Player $i$ instead of $V_i$.
	\item We use $R$ or $r$ to denote the number of round instead of $L$ or $l$.
\end{itemize}

This proof is organized into four steps.

\begin{enumerate}
	\item We first convert our definition of OSSBG into an equivalent extensive-form game with imperfect information in Appx.~\ref{app-sec:to-efg}. We will base our proof on the EFG-form OSSBG and a corresponding EFG-form of our algorithm.
	\item In Sec.~\ref{app-sec:construct-oracle}, we show how to use the prior space approximation and a fixed-prior game solver to build an all-prior game solver (an oracle).
	\item In Sec.~\ref{app-sec:one-step}, we show how to use a subforest induced game oracle and no-regret learning algorithms to solve the game.
	\item In Sec.~\ref{app-sec:full}, we combine the steps above and show how to solve the game from scratch using backward induction.
\end{enumerate}

\subsection{Construct an EFG-form OSSBG}\label{app-sec:to-efg}

\begin{algorithm}[t]
\caption{Construct an extensive form OSSBG}\label{app-alg:cons-full}
\begin{algorithmic}[1]
\Procedure{ConstructSubforest}{$r,s,\boldsymbol{v}$} \Comment{\parbox[t]{.35\linewidth}{Depth $r$, initial state $s$, initial nodes $\boldsymbol{v}=\{v_\lambda\}_{\lambda\in\Lambda}$}}
		
		\If{$r\geq 1$}
				\State Make each node $v_\lambda\in \boldsymbol{v}$ a player 1's decision node
				\For{$a_1\in \mathcal{A}_1$}
					\State Create $|\Lambda|$ new nodes $\{\tilde{v}_{\lambda,a_1}\}_{\lambda\in\Lambda}$
					\State Connect each $v_\lambda$ to $\tilde{v}_{\lambda,a_1}$ via player 1's action $a_1$
					\State Make $\tilde{v}_{\lambda,a_1}$ an player 2's decision node
					\For{$a_2\in \mathcal{A}_2$}
						\State Create $|\Lambda|$ new nodes $\{\tilde{v}_{\lambda, a_1, a_2}\}_{\lambda\in\Lambda}$
						\State Connect each $\tilde{v}_{\lambda, a_1}$ to $\tilde{v}_{\lambda,a_1, a_2}$ via player 2's action $a_2$
						\State Make $\tilde{v}_{\lambda,a_1, a_2}$ a Chance node
						\State $u_i(\tilde{v}_{\lambda,a_1, a_2})\gets \bar{u}_i^\lambda(s,a_1,a_2)$ for $i=1,2$
						\For{$s'\in S$}
							\State Create $|\Lambda|$ new nodes $\{\tilde{v}_{\lambda, a_1, a_2, s'}\}_{\lambda\in\Lambda}$
							\State Connect each $\tilde{v}_{\lambda,a_1, a_2}$ to $\tilde{v}_{\lambda, a_1, a_2, s'}$ with  $P(s'|s, a_1, a_2)$
							\State \Call{ConstructSubforest}{$r-1, s', \{\tilde{v}_{\lambda, a_1, a_2, s'}\}_{\lambda\in\Lambda}$}
						\EndFor
					\EndFor
				\EndFor
		\EndIf
	\EndProcedure
	
	\State
\Function{Construct}{$r, s, p$}\Comment{\parbox[t]{.35\linewidth}{Depth $r$, initial state $s$, initial type distribution $p$}}
	\State Create a Chance node $c$
	\State Create $|\Lambda|$ new nodes $\boldsymbol{v}=\{v_{\lambda}\}_{\lambda\in\Lambda}$
	\State Connect $c$ to each node $v_\lambda\in\boldsymbol{v}$ with probability $p_\lambda$
	\State \Call{ConstructSubforest}{$r, s, \boldsymbol{v}$} 
\EndFunction
\end{algorithmic}
	
\end{algorithm}

\begin{definition}[One-sided Stochastic Bayesian Game]
	We call a game that can be constructed by Algorithm \ref{app-alg:cons-full} an One-sided Stochastic Bayesian Game, noted $\Gamma(r,s,p)$ or simply $\Gamma$. We call $r$ the depth of the game, denoted $\dep(\Gamma)$.
\end{definition}

\subsection{Constructing an oracle given samples}\label{app-sec:construct-oracle}

\begin{definition}[Oracle]
	An oracle $\mathcal{O}_r$ of a game with depth $r$ is defined as a function that takes a state $s$ and prior $p$ as input and outputs a strategy for both player for game $\Gamma(r,s,p)$.
\end{definition}

\begin{lemma}\label{lem:approx}
	Let $\sigma^*$ be an $\epsilon$-NE for an One-sided Stochastic Bayesian Game noted $\Gamma(r,s,p_1)$, then $\sigma^*$ is an $(2dU+\epsilon)$-NE for game $\Gamma(r,s,p_2)$ if $\|p_1-p_2\|_{l_1}\leq d$ and $U=\max_{i\in\{1,2\}, z\in Z}u_i(z)-\min_{i\in\{1,2\}, z\in Z}u_i(z)$.
\end{lemma}

\begin{proof}
	Let the utility function of $\Gamma(r, s,p)$ be $u(p, \sigma)$. We have
	\begin{align}
		u_i(p, \sigma)=&\sum_{z\in Z} \pi^\sigma(z)u_i(z)\\
		=&\sum_{\lambda\in\Lambda} p_\lambda\sum_{z\in Z(v_\lambda)}\pi^\sigma(v_\lambda\rightarrow z) u_i(z) \label{equ:decompose}
	\end{align}

	Let $\hat{u}^\lambda_i(\sigma):=\sum_{z\in Z(v_\lambda)}\pi^\sigma(v_\lambda\rightarrow z)u_i(z)$. Let $\hat{\boldsymbol{u}}_i(\cdot):=[\hat{u}^1_i(\cdot),\dots, \hat{u}^{|\Lambda|}_i(\cdot)]^\top$. From (\ref{equ:decompose}), we can see that $u_i(p,\sigma)=p\cdot \hat{\boldsymbol{u}}_i(\sigma)$.

	Since $\sigma^*$ is an $\epsilon$-NE for $\Gamma(r,s,p_1)$, we have
	
	\begin{align}
    	\max_{\sigma'_i\in\Sigma_i}u_i(p_1,(\sigma'_i, \sigma^*_{-i}))-u_i(p_1,\sigma^*)&\leq \epsilon\\
    	\max_{\sigma'_i\in\Sigma_i}p_1\cdot \hat{\boldsymbol{u}}_i(\sigma'_i, \sigma^*_{-i})-p_1\cdot \hat{\boldsymbol{u}}_i(\sigma^*)&\leq \epsilon
	\end{align}
	
	Therefore, using $\sigma^*$ in $\Gamma(r,s,p_2)$ yields
	
	\begin{align}
		&\max_{\sigma'_i\in\Sigma_i}u_i(p_2,(\sigma'_i,\sigma^*_{-i}))-u_i(p_2,\sigma^*)\\
		=&\max_{\sigma'_i\in\Sigma_i}p_2\cdot \hat{\boldsymbol{u}}_i(\sigma'_i, \sigma^*_{-i})-p_2\cdot \hat{\boldsymbol{u}}_i(\sigma^*)\\
		=&\max_{\sigma'_i\in\Sigma_i}(p_2-p_1)\cdot \hat{\boldsymbol{u}}_i(\sigma'_i, \sigma^*_{-i})-(p_2-p_1)\cdot \hat{\boldsymbol{u}}_i(\sigma^*)\ +\nonumber\\
		&\qquad p_1\cdot \hat{\boldsymbol{u}}_i(\sigma'_i, \sigma^*_{-i})-p_1\cdot \hat{\boldsymbol{u}}_i(\sigma^*)\\
		\leq&\max_{\sigma'_i\in\Sigma_i}(p_2-p_1)\cdot (\hat{\boldsymbol{u}}_i(\sigma'_i, \sigma^*_{-i})-\hat{\boldsymbol{u}}_i(\sigma^*))\ +\nonumber\\
		&\qquad \max_{\sigma'_i\in\Sigma_i}p_1\cdot \hat{\boldsymbol{u}}_i(\sigma'_i, \sigma^*_{-i})-p_1\cdot \hat{\boldsymbol{u}}_i(\sigma^*)\\
		\leq& dU+\epsilon
	\end{align}
\end{proof}

\begin{lemma}
	A weighted sum of multiple $\epsilon$-NEs is also an $\epsilon$-NE.
\end{lemma}

\begin{proof}
	Let $\Sigma^*=\{\sigma^*_1,\dots,\sigma^*_n\}$ be a set of $\epsilon$-NEs for some game with utility function $u$. Let $\boldsymbol{w}$ be a weight distribution over $\Sigma^*$, where $\sum_{j=1}^n \boldsymbol{w}_j=1$ and $\boldsymbol{w}_j\geq 0$ for all $j=1,\dots, n$.
	
	Let the weighted mixed strategy be $\sigma^*=\sum_{j=1}^n\boldsymbol{w}_j\sigma^*_j$. For each player $i$, We have
	
	\begin{align}
		&\max_{\sigma'_i\in\Sigma_i}u_i(\sigma'_i, \sigma^*_{-i})	-u_i(\sigma^*)\\
		=&\max_{\sigma'_i\in\Sigma_i}\sum_{j=1}^n\boldsymbol{w}_i\Big(u_i(\sigma'_i,\sigma^*_{j,-i})-u_i(\sigma^*_j)\Big)\\
		\leq&\sum_{j=1}^n\boldsymbol{w}_i\Big(\max_{\sigma'_i\in\Sigma_i}u_i(\sigma'_i,\sigma^*_{j,-i})-u_i(\sigma^*_j)\Big)\\
		\leq&\sum_{j=1}^n\boldsymbol{w}_i\epsilon\\
		=&\epsilon
	\end{align}

\end{proof}

\begin{definition}[$d$-dense]
	We call a type distribution set $\boldsymbol{p}=\{p_1,\dots, p_K\}$ $d$-densely distributed if for any point $p$ in the type distribution space, there is at least one point $p^*$ in $\boldsymbol{p}$ s.t. $\|p-p^*\|\leq d$.
\end{definition}

\begin{algorithm}
\caption{Approximate NE for $\Gamma(r,s,\cdot)$}\label{app-alg:approx}
\begin{algorithmic}[1]
\Function{Prepare}{$d$}\Comment Density parameter $d$	\State Choose a $d$-densely distributed set $\boldsymbol{p}=\{p_1,\dots, p_K\}$
	\For {$i:=1\rightarrow K$}
		\State Compute an $\epsilon$-NE $\sigma^*_i$ for $\Gamma(r,s,p_i)$
	\EndFor
	\State $\boldsymbol{\sigma}^* \gets \{\sigma^*_1,\dots, \sigma^*_K\}$
	\State \Return $\boldsymbol{p}, \boldsymbol{\sigma}^*$
\EndFunction
\Function{Approximate}{$\boldsymbol{p}, \boldsymbol{\sigma}^*, p$}
	\State $\boldsymbol{i}\gets\{i| \|p_i-p\|_{l_1}\leq d\}$
	\State \Return a weighted sum of $\{\sigma^*_{i}|i\in\boldsymbol{i}\}$
\EndFunction
\end{algorithmic}

\end{algorithm}

\begin{corollary}
	The approximated NE in Algorithm~\ref{app-alg:approx} is $(dU+\epsilon)$-NE.
\end{corollary}

\subsection{Solve one game with subforest induced game oracles}\label{app-sec:one-step}

\begin{definition}[Subforest]
	We call the recursive structure (a collection of sub-trees) constructed by a call of \textsc{ConstructSubforest} during a construction of an OSSBG a subforest of that game, denoted as $F$. We call the parameter $r$ the depth of the subforest, denoted as $\dep(G)$. Given a game $\Gamma$, a subforest $F$ of $\Gamma$ can be uniquely identified by the action-state history $h=\{s^0,\boldsymbol{a}^0, s^1,\dots,\boldsymbol{a}^{n-1}, s^n\}$.
\end{definition}

Due to the recursive structure of a subforest, we can choose any subforest $G$ then call a \textsc{Construct} operation at the top of it with a type distribution parameter $p$ to form a new OSSBG. We call this game a subforest $G$-induced game, denoted $\Gamma(F, p)$.

We call a subforest $F$ of $\Gamma$ a first-level subforest iff $\dep(F)=\dep(\Gamma)-1$. $F$ can be uniquely identified by the first action-state pair $\boldsymbol{a}^0,s^0$.


We can decompose a full strategy $\sigma$ for game $\Gamma(r,s,p)$ into $1+|\Omega|\cdot|\mathcal{A}_1|\cdot|\mathcal{A}_2|$ strategies: $\hat{\sigma}$, which represents the strategies for the first one action for both players; and $\sigma^{a_1, a_2, s'}$, which represents the strategies for all the first-level subforests $G(a_1,a_2,s')$. 

Suppose we have an oracle for all first-level subforests induced games, we now only need to focus on finding a $\epsilon$-NE for the first step strategies $\hat{\sigma}$. We formulate this problem into an  \textbf{online convex programming problem}.

\begin{definition}
An \textbf{online convex programming problem} consists of a strategy space $\Sigma\subseteq\mathbb{R}^n$ and an infinite sequence $\{c^1, c^2,\dots\}$ where each $c^t:\Sigma\mapsto\mathbb{R}$ is a convex function.

At each step $t$, the \textbf{online convex programming algorithm} selects an vector $x^t\in \Sigma$. After the vector is selected, it receives the payoff function $c^t$.

\end{definition}

The objective of an online convex programming problem is to propose a sequence of $x^t$ that maximizes $\lim_{T\rightarrow \infty}\frac{1}{T}\sum_{t=1}^Tc^t(x^t)$.

In particular, for Player 1, the strategy space is a concatenation of all $|\Lambda|$ types' corresponding strategies, i.e. $\Sigma_1=(\Delta_{\mathcal{A}_1})^{|\Lambda|}$; for Player 2, since it doesn't have the type information, the strategy space is simply $\Sigma_2=\Delta_{\mathcal{A}_2}$. Finding an optimal solution can be modeled into two parallel online convex programming problems, where both players aim to maximize their own payoff and can influence each other's payoff at each time step. The time-averaged strategy $\frac{1}{T}\sum_{t=1}^T x^t$ is used as the optimal solution.

Modeling a multi-player optimizing problem into an online programming problem is a common technique in equilibrium-attaining literature.

Another important concept associated with online convex programming problem in \textit{external regret}.

\begin{definition}[External regret]
	The external regret of an online convex programming algorithm is defined as
	\begin{equation}
		R^T=\frac{1}{T}\max_{x^*\in \Sigma}\sum_{t=1}^Tc^t(x^*)-c^t(x^t)
	\end{equation}
\end{definition} 

An online convex programming algorithm is said to be \textit{no-regret} iff for any $\{c^t\}$, $\lim_{T\rightarrow\infty}R^T=0$.

In our paper, we consider two streams of such algorithms: \textit{gradient ascent} and \textit{regret matching}.

\subsubsection*{Gradient ascent}

\begin{align}
	\theta^{t+1}&=\theta^t+\eta_t\nabla c^t(x^t)\\
	x^{t+1}&=\textsc{SoftMax}(\theta^{t+1})
\end{align}

where $\eta_t$ is the learning rate. With a carefully chosen learning rate sequence, this learning algorithm is no-regret. In some complicated games, we use a sample-based version of the algorithm and a parameterized policy.

This stream corresponds to the TISP-PG algorithm. Let $\|\nabla c\|=\sup_{x\in \Sigma, t=1,2,\dots}\|\nabla c^t(x)\|$. Gradient ascent with $\eta_t=t^{-1/2}$ has a regret bound~\cite{Zinkevich2003OnlineCP} of 

\begin{equation}
	\epsilon_{\text{min}}(T)=\frac{\|\nabla c\|^2}{\sqrt{T}}+o\left(\frac{1}{\sqrt{T}}\right)
\end{equation}

\subsubsection*{Regret matching}

Let $r^{t+1}(a)=\sum_{\tau=1}^t c^\tau (a)-c^\tau(x^\tau)$,

\begin{equation}
	x^{t+1}(a)=\frac{(r^{t+1}(a))^+}{\sum_{a'}(r^{t+1}(a'))^+}
\end{equation}

where $(\cdot)^+:=\max(\cdot, 0)$. 

Regret matching is an algorithm that has been rapidly gaining popularity over the years in solving imperfect information games, due to its success in Counterfactual Regret Minimization. Compared to gradient ascent, regret matching is hyperparameter-free and requires less computational resources since it doesn't need access to the gradient. However, combining regret matching with function approximators such as neural networks can be challenging. Nonetheless, we take inspirations from Deep CFR and propose a variant of regret matching that is sample-based and parameterizable for comparison purposes. 

This stream corresponds the TISP-CFR algorithm. Let $C=\sup_{x\in \Sigma, t=1,2,\dots}c^t(x)-\inf_{x\in \Sigma, t=1,2,\dots}c^t(x)$. Regret matching has a regret bound~\cite{zinkevich2008regret} of 

\begin{equation}
	\epsilon_{\text{min}}(T)=\frac{C\sqrt{|\mathcal{A}|}}{\sqrt{T}}+o\left(\frac{1}{\sqrt{T}}\right)
\end{equation}

We call an no-regret online convex programming algorithm a \textit{regret minimizing device} or a \textit{regret minimizer}.

We now introduce an algorithm that focuses on finding a NE for $\hat{\sigma}$ with the help of a regret minimizer and an first-level subforest induced game oracle. See algorithm \ref{app-alg:first}.

\begin{algorithm}[H]
	\caption{Compute first-level NE with regret minimizing device and first-level oracle for $\Gamma(r, s, p)$}\label{app-alg:first}

	\begin{algorithmic}[1]
		\Function{ComputeFirstLevel}{$T$}\Comment Number of iterations $T$
			\For{$t:=1\rightarrow T$}
				\State Let $\hat{\sigma}^t_1=\{\hat{\sigma}^t_{1,\lambda}\}, \hat{\sigma}^t_2$ be the strategies produced by the regret minimizing device
				\For{$a_1\in\mathcal{A}_1$}
					\For{$a_2\in\mathcal{A}_2$}
						\State $u^{\text{tmp}}_i\gets \sum_{\lambda}p_\lambda u^\lambda_i(s,a_1,a_2)$ for $i=1,2$
						\If {$r>1$}
							\For {$s'\in\Omega$}
								\State Let $F$ be the first-level subforest identified by $(a_1,a_2), s'$
								\State Let the roots of $F$ be $\boldsymbol{v}=\{v_\lambda\}_{\lambda\in\Lambda}$
								\State $p'_\lambda\gets \frac{p_\lambda\hat{\sigma}^t_{1,\lambda}(a_1)}{\sum_{\lambda'}p_{\lambda'}\hat{\sigma}^t_{1,\lambda'}(a_1)}$ \Comment The new belief
								\State $p'\gets \{p'_\lambda\}_{\lambda\in\Lambda}$
								\State Compute a $\epsilon$-NE for $\Gamma(F, p')$ and let the equilibrium value be $u^*_1, u^*_2$
								\State $u^{\text{tmp}}_i\gets u^{\text{tmp}}_i+P(s'|s,a_1,a_2)u^*_i$ for $i=1,2$
							\EndFor
							\EndIf
						\State $u^t_i(a_1,a_2)\gets u^{\text{tmp}}_i$ for $i=1,2$
					\EndFor
				\EndFor
				\State Report $u^t_1, u^t_2$ to the regret minimizing device
			\EndFor
		\EndFunction	
	\end{algorithmic}

\end{algorithm}

Suppose the regret minimizing device minimizes the regret in the rate of $\epsilon_{\text{min}}(T)$, i.e., 

\begin{equation}
	\frac{1}{T}\max_{\sigma'_i\in \Sigma_i}\sum_{t=1}^T\big(u^t_i(\sigma'_i, \sigma^t_{-i})-u_i^t(\sigma^t)\big)\leq \epsilon_{\text{min}}(T)
\end{equation}

We can see that Algorithm \ref{app-alg:first} is equivalent of playing in the full game where the full strategy is a concatenation of $\hat{\sigma}^t$ and the $\epsilon$-NE strategies in each iteration and each subforest. We denote the concatenated strategies $\sigma^t$.

The overall regret of the full game for Player $i$ is:

\begin{align}
R_i^T=&\frac{1}{T}\max_{\sigma'_i\in\Sigma_i}\sum_{t=1}^T\big(u_i(\sigma'_i, \sigma^t_{-i})-u_i(\sigma^t)\big)	\\
=&\frac{1}{T}\max_{a_i\in \mathcal{A}_i}\max_{\sigma'_i\in\Sigma_i}\sum_{t=1}^T\big(u_i(\sigma^t_i|_{a^0_i= a_i}, \sigma^t_{-i})-u_i(\sigma^t)\nonumber\\
&+\mathbb{E}_{a_{-i}, s'}[u_i(\sigma^t_i|_{F(\boldsymbol{a}, s')\rightarrow \sigma'_i}, \sigma^t_{-i})-u_i(\sigma^t_i|_{a^0_i= a_i}, \sigma^t_{-i})]\big)\\
\leq & \frac{1}{T}\max_{a_i\in \mathcal{A}_i}\sum_{t=1}^T\big(u_i(\sigma^t_i|_{a^0_i= a_i}, \sigma^t_{-i})-u_i(\sigma^t)\big)+\nonumber\\
&\quad \frac{1}{T}\max_{a_i\in \mathcal{A}_i}\sum_{t=1}^T\mathbb{E}_{a_{-i}, s'}[{\color{red}\max_{\sigma'_i\in\Sigma_i}\big(u_i(\sigma^t_i|_{F(\boldsymbol{a}, s')\rightarrow \sigma'_i}, \sigma^t_{-i})}\nonumber\\
&\quad{\color{red}-u_i(\sigma^t_i|_{a^0_i= a_i}, \sigma^t_{-i})\big)}]\label{equ:regret-decompose}
\end{align}

The left adding part of (\ref{equ:regret-decompose}) is the regret minimized in Algorithm \ref{app-alg:first}; The red part of (\ref{equ:regret-decompose}) is the maximum profit player $i$ can make when deviating from an $\epsilon$-NE in a first-level subforest. Therefore,

\begin{align}
	R_i^T\leq & \epsilon_{\text{min}}(T)+\frac{1}{T}\max_{a_i\in \mathcal{A}_i}\sum_{t=1}^T\mathbb{E}_{a_{-i}, s'}[\epsilon]\\
	=&\epsilon_{\text{min}}(T)+\epsilon
\end{align}

Therefore, $\bar{\sigma}=\frac{1}{T}\sum_{t=1}^T\sigma^t$ is an $(\epsilon_{\text{min}}(T)+\epsilon)$-CCE. If the game is zero-sum, then $\bar{\sigma}_i=\frac{1}{T}\sum_{t=1}^T\sigma^t_i$ is an $(\epsilon_{\text{min}}(T)+\epsilon)$-NE.

\subsection{Solve the full game by induction}\label{app-sec:full}

\begin{algorithm}
\caption{Solve the full game $\Gamma(R,s_0,p_0)$}\label{app-alg:full}
\begin{algorithmic}[1]
	\Function{ComputeOracle}{$r,\mathcal{O}_{r-1}$}
		\State Set up the \textsc{ComputeFirstLevel} with $r,T,\mathcal{O}_{r-1}$
		\State $\boldsymbol{p}, \boldsymbol{\sigma}^*\gets$\Call{Prepere}{$d$}
		\State $\mathcal{O}_r\gets$\Call{Approximate}{$\boldsymbol{p}, \boldsymbol{\sigma}^*, \cdot$}
		\State \Return $\mathcal{O}_r$
	\EndFunction
	

	\Function{Solve}{$d, T, R$}
		\State $\mathcal{O}_0\gets\text{Arbitrary}$
		\For{$r\gets 1,\dots R$}
			\State $\mathcal{O}_r\gets$\Call{ComputeOracle}{$r,\mathcal{O}_{r-1}$}
		\EndFor
		\State \Return $\mathcal{O}_1,\dots,\mathcal{O}_R$
	\EndFunction
\end{algorithmic}
	
\end{algorithm}

\begin{lemma}\label{lem:oracle}
	Strategies produced by oracle $\mathcal{O}_r$ is $\epsilon$-NE for $r=1,\dots, R$, where $\epsilon=r(dU+\epsilon_{\text{min}}(T))$.
\end{lemma}

\begin{proof}
We prove by induction. We expand the lemma to include $r=0$. Since $\mathcal{O}_0$ will not be called in Algorithm \ref{app-alg:first}, it is fine to assume that it produces $0$-NE.

Let $\bar{\epsilon}=dU+\epsilon_{\text{min}}(T)$. Suppose $r\geq1$, by induction, we know that $\mathcal{O}_{r-1}$ produces $(r-1)\bar{\epsilon}$-NEs. According to Sec.~\ref{app-sec:one-step}, the strategies produced by Algorithm \ref{app-lg:first} is $(\epsilon_{\text{min}}(T)+(r-1)\bar{\epsilon})$-NE. Therefore, the oracle $\mathcal{O}_r$ created by Algorithm \ref{app-alg:approx} produces $(\epsilon_{\text{min}}(T)+(r-1)\bar{\epsilon}+dU)=r\bar{\epsilon}$-NEs. Thus the induction.
\end{proof}

%
%

\begin{algorithm}
\caption{Convert from oracles to history-based strategy}\label{app-alg:convert}

\begin{algorithmic}[1]
	\Function{GetStrategy}{$h^r, \mathcal{O}_1,\dots,\mathcal{O}_R$}
		\State $h^r=\{s^0,(a_1^0,a_2^0),\dots,(a_1^{r-1},a_2^{r-1}),s^r\}$
		\For{$i\gets 0,\dots, r-1$}
			\State $\sigma^i_1\gets \mathcal{O}_{R-i}(s^i, p^i)$
			\State $p^{i+1}\gets \textsc{UpdateBelief}(p^i, \sigma^i_1, a_1^i)$
		\EndFor
		\State $\sigma_1, \sigma_2\gets \mathcal{O}_{R-r}(s^r, p^r)$
		\State \Return $\sigma_1, \sigma_2$
	\EndFunction
\end{algorithmic}
	
\end{algorithm}

\begin{corollary}
	Algorithm \ref{app-alg:convert} produces $R(dU+\epsilon_{\text{min}}(T))$-PBE strategies.
\end{corollary}

